\documentclass[11pt, letterpaper]{article}

\usepackage{graphicx}
\usepackage{epstopdf}
\usepackage{amsmath,amssymb,amsfonts}
\usepackage{setspace}
\usepackage{epstopdf}
\usepackage{color}
\usepackage[letterpaper,left=1in,right=1in,top=1in,bottom=1in]{geometry}
\usepackage{multirow}
\usepackage{longtable}
\usepackage{rotating}
\usepackage{mathrsfs}
\usepackage{algorithm}
\usepackage{algorithmicx}
\usepackage{algpseudocode}
\usepackage{booktabs}
\newtheorem{rmk}{Remark}
\newtheorem{assumption}{Assumption}

\newtheorem{lemma}{Lemma}

\newtheorem{theorem}{Theorem}

\newtheorem{definition}{Definition}
\newtheorem{proof}{Proof of Theorem}

\title{\Large \bf Iterative distributed moving horizon estimation of linear systems with penalties on both system disturbances and noise}

\author{
\centerline{\normalsize Xiaojie Li$^{a}$, Song Bo$^{b}$, Yan Qin$^{c}$, Xunyuan Yin$^{a,}$\thanks{Corresponding author: X. Yin. Tel: (+65) 6316 8746. Email: xunyuan.yin@ntu.edu.sg.}
}
\vspace{5mm}\\
\centerline{\small $^{a}$School of Chemistry, Chemical Engineering and Biotechnology, Nanyang Technological University,}\\
\centerline{\small 62 Nanyang Drive, 637459, Singapore}\\
\centerline{\small $^{b}$Department of Chemical \& Materials Engineering, University of Alberta,}\\
\centerline{\small Edmonton, AB T6G 1H9, Canada}\\
\centerline{\small $^{c}$School of Electrical and Electronic Engineering, Nanyang Technological University,}\\
\centerline{\small 50 Nanyang Avenue, 639798, Singapore}
}
\allowdisplaybreaks
\begin{document}

\date{}
\maketitle
\setstretch{1.45}

\begin{abstract}
In this paper, partition-based distributed state estimation of general linear systems is considered. A distributed moving horizon state estimation scheme is developed via decomposing the entire system model into subsystem models and partitioning the global objective function of centralized moving horizon estimation (MHE) into local objective functions. The subsystem estimators of the distributed scheme that are required to be executed iteratively within each sampling period are designed based on MHE. Two distributed MHE algorithms are proposed to handle the unconstrained case and the case when hard constraints on states and disturbances, respectively. Sufficient conditions on the convergence of the estimates and the stability of the estimation error dynamics for the entire system are derived for both cases. A benchmark reactor-separator process example is introduced to illustrate the proposed distributed state estimation approach.
\end{abstract}

\noindent{\bf Keywords:} Distributed state estimation, partition-based framework, moving horizon estimation (MHE), iterative evaluation

\section{Introduction}
The manufacturing industry has undergone a significant transformation in recent years, marked by a rapid increase in the scale, complexity, and level of integration of manufacturing systems and industrial processes. As a result, the conventional centralized and decentralized decision-making paradigms are no longer sufficient to meet the demands of the industries for advanced control solutions for safe, flexible, and sustainable process operation, and consistent and profitable production \cite{christofides2013distributed, daoutidis2018integrating, daoutidis2016sustainability}. This trend has highlighted the necessity of using distributed architecture to develop next-generation flexible and scalable decision-making solutions for large-scale industrial systems with tightly integrated physical components. Within a distributed framework, multiple computing agents are deployed to perform process monitoring and/or collaboratively make decisions on process operation through real-time communication. This way, the major advantages of the centralized and the decentralized methods are inherited, while the associated limitations are bypassed. In particular, a higher level of fault tolerance, maintenance flexibility as well as computational efficiency can be achieved as compared to the centralized methods, and the overall performance will be significantly enhanced as compared to the structurally flexible decentralized counterpart due to more appropriate treatment of the dynamic interactions which are in general overlooked in decentralized designs \cite{christofides2013distributed, daoutidis2016sustainability, yin2018forming}.

Distributed state estimation and distributed control are two main components of a complete distributed decision-making system. It is worth pointing out that, while more research attention has been attached to the development of new distributed control algorithms (in particular, distributed model predictive control in process control \cite{liu2009distributed, chen2022barrier, chen2021cyber, tang2021coordinating, pourkargar2019distributed}), the dual problem $-$ distributed state estimation $-$ is equivalently important. The multiple local estimators of distributed state estimation take advantage of online measurements from the available hardware sensors to collaboratively infer full-state information for a process \cite{battistelli2016stability, haber2013moving, liu2022distributed, yin2020distributed}; good estimates provided by distributed state estimation are critical for distributed control to make informed decisions for appropriate operation intervention. Distributed state estimation methods can be divided into two categories: 1) consensus-based distributed estimation, and 2) partition-based distributed estimation. Within a consensus-based framework, local estimators are typically designed based on a global model, and they estimate the same quality variables of the entire system \cite{yin2021consensus, olfati2007distributed, farina2012distributed, das2016consensus+, battistelli2018distributed}. While in a partition-based design, the local estimators are designed based on decomposed subsystem models, and each estimator accounts for an exclusive subset of state variables (i.e., the states of the corresponding subsystem model). Partition-based distributed state estimation can provide improved fault tolerance, computational efficiency, and organizational flexibility of advanced process control systems, and it is investigated in the current work.

There have been some results on partition-based state estimation in the existing literature \cite{farina2009moving, schneider2015iterative, vadigepalli2003distributed}. In \cite{yin2018forming, zhang2013distributed, yin2017distributed}, complex nonlinear processes were decomposed into smaller interacting subsystems, based on which local estimators were designed. These estimators are evaluated in a non-iterative manner (that is, the estimation algorithm is executed only once as new measurements are made available). While this type of method further facilitates the reduction in the usage of computing resources, the estimation results may be noticeably less accurate than those of the centralized counterpart. In \cite{schneider2015convergence, schneider2017solution, yin2022event}, distributed moving horizon estimation (DMHE) algorithms that require iterative executions of the local estimators within each sampling time were proposed, such that the estimates may converge to the estimates of linear centralized MHE. In particular, in \cite{schneider2015convergence}, the objective function of centralized MHE is decomposed into multiple individual functions. Then, each decomposed individual function is assigned a sensitivity term to form a subsystem's objective function for local MHE. By selecting appropriate local estimator parameters, the distributed estimation scheme can be made asymptotically stable, and the state estimates can converge to those of the centralized MHE. In \cite{schneider2017solution}, an iterative partition-based DMHE algorithm that can be used to develop stable distributed estimation schemes based on any configurations of the subsystem models was proposed. More relevant results on partition-based distributed state estimation with iterative executions of the local estimators can be found in \cite{farina2009moving, schneider2015iterative, farina2010moving}, and the references therein.
The major advantage of these iterative distributed state estimation approaches lies in that they have the potential to provide estimates convergent to the centralized counterpart; this is favorable when more accurate estimates or faster convergence are needed \cite{schneider2015convergence, schneider2017solution}.

While the distributed estimation approaches with iterative evaluations of the local estimators have their inherent advantages over the non-iterative counterparts, in most of the existing iterative designs (including the representative methods in \cite{schneider2015convergence, schneider2017solution}), penalty on system disturbances is absent from the local objective functions, leading to difficulty in handling constraints on system disturbances which are common in real processes. It is worthwhile to consider these disturbances when formulating the designs of the local estimators of a DMHE scheme. In \cite{farina2010moving}, an initial attempt was made on incorporating system disturbances into the objective functions of MHE-based estimators. In this approach, each MHE-based estimator only receives and utilizes the sensor measurements of the corresponding local subsystem, and the estimators are evaluated in a non-iterative fashion.
In \cite{schneider2013iterative}, penalty on system disturbances was also included in the objective function of each MHE-based estimator, and the convergence and the optimality of the proposed partition-based distributed estimation algorithm were proven. At the same time, it is noted that \cite{schneider2013iterative} did not address the constrained case, and the stability of the estimation error dynamics was not studied.

Based on the above observations, we aim to extend the approach in \cite{schneider2015convergence} to develop a partition-based iterative DMHE scheme for medium- to large-scale linear systems, of which the local objective functions for the subsystem estimators penalize both system disturbances and measurement noise for guaranteed estimation performance. Moving horizon estimation converts the state estimation problem into optimization, such that optimal estimates can be obtained in the presence of various constraints on variables \cite{haseltine2005critical, rao2003constrained, liu2013moving, https://doi.org/10.48550/arxiv.2206.10397, valipour2022extended}. In this work, the entire system is partitioned into subsystems that interact with each other. An individual objective function incorporating penalties on both subsystem disturbances and measurement noise is formulated, and local MHE-based estimators are developed to provide estimates of the subsystem states in a collaborative manner. Both the unconstrained case and constrained case (when hard constraints on system states and/or disturbances) are addressed by proposing two different DMHE formulations. The stability of the entire distributed estimation scheme is proven under the two case scenarios. The proposed method is illustrated using a simulated chemical process, and the results confirm the effectiveness and superiority of the proposed method.
Partial preliminary results of this work were submitted as a conference paper \cite{IFAC}. In addition to the results on unconstrained DMHE reported in the conference version \cite{IFAC}, the current paper also presents the DMHE formulation for the constrained case,  and proves the stability of estimation error dynamics. Accordingly, this paper presents the simulation results for the proposed method considering the case when constraints on decision variables are present. Also, as compared to \cite{IFAC}, this paper presents extended explanations and discussions of the proposed algorithms, and more detailed derivations for the stability results.

\section{Preliminaries}
\subsection{Notation}
$\lambda_{\max}(A)$ and $\lambda_{\min}(A)$ denote the largest and the smallest eigenvalues of matrix $A$, respectively. ${\rm{col}}(x_1, x_2, \ldots, x_n)$ is a column vector consisting of $x_1, x_2, \ldots, x_n$. $\{x\}_{a}^{b}$ represents a column vector $[x_{a}, x_{a+1},\ldots, x_{b}]$. The operators $\|\cdot\|$ and $\|\cdot\|_{P}$ represent Euclidean norm and the weighted Euclidean norm of a scalar or a vector, respectively; they are computed as $\|x\|^2=x^{\mathrm{T}}x$ and $\|x\|^2_{P}=x^{\mathrm{T}}Px$. For scalars $a_i$, $i=1,2,\ldots,n$, $\text{diag}\left\{a_1,a_2,\ldots,a_n\right\}$ is a diagonal matrix where the elements on the main diagonals are $a_i$, $i=1,\ldots,n$. For matrices $A_i$, $i=1,\ldots,n$, $\text{diag}\left\{A_1,A_2,\ldots,A_n\right\}$ is a block diagonal matrix. $[A_{ij}]$ is a block matrix where $A_{ij}$ is the submatrix in the $i$th row and $j$th column. $[x]^{+}$ denotes the orthogonal projection of a vector $x$ onto the convex set $\mathbb{X}$, which is computed by $[x]^{+} = \mathrm{arg}\min_{y\in \mathbb{X}}\|y-x\|$.

\subsection{System description}
We consider general linear systems that can be partitioned into $n$ interconnected subsystems. The model of subsystem $i$, $i\in \mathbb{N}=\{1,2,\ldots,n\}$, is described by the following state-space form:
\begin{subequations}
  \begin{align}
   x_{k+1}^{i}&=A_{ii}x_{k}^{i}+\sum_{j\in\mathbb{N}\setminus\{i\}}A_{ij}x_{k}^{j}+B_{ii}u_{k}^{i}+\sum_{j\in\mathbb{N}\setminus\{i\}}B_{ij}u_{k}^{j}+w_{k}^{i}\label{modeli1}\\
   y_{k}^{i}&=C_{ii}x_{k}^{i}+v_{k}^{i}\label{modeli2}
\end{align}
\end{subequations}
where $x_{k}^{i}\in\mathbb{R}^{n_{x^i}}$ is the state vector for subsystem $i$ at time instant $k$; $u_{k}^{i}\in\mathbb{R}^{n_{u^i}}$ is the vector of known inputs which consist of the control inputs and/or known disturbances for subsystem $i$; $y_{k}^{i}\in\mathbb{R}^{n_{y^i}}$ is the vector of sensor measurements for subsystem $i$; $w_{k}^{i}\in\mathbb{R}^{n_{x^i}}$ and $v_{k}^{i}\in\mathbb{R}^{n_{y^i}}$ are the vectors of system disturbances and sensor measurement noise for subsystem $i$, respectively; $A_{ii}$, $A_{ij}$, $B_{ii}$, $B_{ij}$, and $C_{ii}$ are subsystem matrices of compatible dimensions.

\subsection{Centralized moving horizon estimation}\label{section:2.3}
The subsystem models in Equations~\eqref{modeli1} and \eqref{modeli2} constitute the composite model for the entire linear system in the following form:
\begin{subequations}\label{model}
  \begin{align}
   x_{k+1} & =Ax_{k}+Bu_{k}+w_{k} \label{model1}\\
   y_{k} & =Cx_{k}+v_{k} \label{model2}
\end{align}
\end{subequations}
where $x_{k}={\rm{col}}(x_{k}^1, x_{k}^2,\ldots, x_{k}^{n})\in\mathbb{R}^{n_x}$, $u_{k}={\rm{col}}(u_{k}^1, u_{k}^2,\ldots, u_{k}^{n})\in\mathbb{R}^{n_{u}}$, $y_{k}={\rm{col}}(y_{k}^1, y_{k}^2,\ldots, y_{k}^{n})\in\mathbb{R}^{n_{y}}$, $w_{k}={\rm{col}}(w_{k}^1, w_{k}^2,\ldots, w_{k}^{n})\in\mathbb{R}^{n_{x}}$ and $v_{k}={\rm{col}}(v_{k}^1, v_{k}^2,\ldots, v_{k}^{n})\in\mathbb{R}^{n_{y}}$. $A=[A_{ij}]$; $B =[B_{ij}]$; $C=[C_{ij}]$ with $C_{ij}=0$ if $i\neq j$, $\forall~i, j\in \mathbb{N}$.


With \eqref{model1} and \eqref{model2}, a centralized moving horizon estimation (CMHE) algorithm, which is an analog of the design in \cite{rao2001constrained}, is reviewed. This CMHE method is the basis for designing the local estimators of the distributed estimation scheme.

At each time instant $k=N, N+1, \ldots$, CMHE solves the following optimization problem
\begin{subequations}\label{CMHE_opt}
\begin{align}
     &\quad\quad\quad\min_{\hat{x}_{k-N|k},~\{\hat{w}\}_{k-N}^{k-1}}\Phi_{k}  \\
     &{\rm{s.t.}}~ \hat{x}_{t+1|k} = A\hat{x}_{t|k}+Bu_{t}+\hat{w}_{t},~t = k-N,\ldots, k-1\\
     &\quad\quad\,\,~ \hat{v}_{t} = y_{t}-C\hat{x}_{t|k}, ~t = k-N ,\ldots, k
  \end{align}
\end{subequations}
%
where the objective function is as follows:
\begin{equation}\label{CMHE_obj}
 \Phi_k  = \frac{1}{2}\Bigg(\|\hat{x}_{k-N|k}-\bar{x}_{k-N}\|_{P^{-1}}^{2}+\sum_{j=k-N}^{k-1}\|\hat{w}_{j}\|^2_{Q^{-1}}+\sum_{j=k-N}^{k}\|\hat{v}_{j}\|^2_{R^{-1}}\Bigg)
\end{equation}
In \eqref{CMHE_obj}, $N\geq 1$ is the length of the estimation horizon, $\hat{x}_{k-N|k}$ is an estimate of the state at time instant $k-N$ calculated at the current time instant $k$, and $\bar{x}_{k-N}$ serves as  {\em a priori} estimate of state $x_{k-N}$, which can be computed following $\bar{x}_{k-N}=A\hat{x}_{k-N-1|k-1}+Bu_{k-N-1}$. $P$, $Q$, and $R$ are weighting matrices that are typically made positive definite.

Before proceeding further, we define three matrices and present the main assumptions of this work as follows:
\begin{align*}
  O & = \left[C^{\mathrm{T}},(CA)^{\mathrm{T}},\ldots,(CA^{N})^{\mathrm{T}}\right]^{\mathrm{T}}, \\
  \Gamma & =
  \left[
    \begin{array}{ccccc}
      0 & 0 & \ldots & 0 & 0 \\
              C & 0 &  & \vdots & \vdots \\
              CA & C & \ddots & \vdots & \vdots\\
              \vdots & \vdots & \ddots & 0 & \vdots \\
              \vdots & \vdots &  & C & 0 \\
              CA^{N-1} & CA^{N-2} & \ldots & CA & C \\
    \end{array}
  \right],~
  \Lambda = \left[\begin{array}{ccccc}
      0 & 0 & \ldots & 0 & 0 \\
              CB & 0 &  & \vdots & \vdots \\
              CAB & CB & \ddots & \vdots & \vdots\\
              \vdots & \vdots & \ddots & 0 & \vdots \\
              \vdots & \vdots &  & CB & 0 \\
              CA^{N-1}B & CA^{N-2}B & \ldots & CAB & CB\\
    \end{array}
  \right].
\end{align*}
\begin{assumption}\label{assume1}
  $w_k$ and $v_k$ are stochastic disturbances and noise of which the statistics are unknown. Moreover, $\mathbb{W}$ and $\mathbb{V}$ that contain $w_k$ and $v_k$, respectively, are compact as described below:
  \begin{align*}
     \mathbb{W} & =\{w_k\in \mathbb{R}^{n_x},~{\rm{s.t}}.~\|w_k\|\leq r_w\} \\
    \mathbb{V} & =\{v_k\in \mathbb{R}^{n_y},~{\rm{s.t}}.~\|v_k\|\leq r_v\}
  \end{align*}
  with $r_w$ and $r_v$ being positive scalars.
\end{assumption}
\begin{assumption}\label{assume2}
The composite system \eqref{model} is observable in $N+1$ steps, that is,  $\left[C^{\mathrm{T}},(CA)^{\mathrm{T}},\ldots,(CA^{N})^{\mathrm{T}}\right]^{\mathrm{T}}$ is full-rank.
\end{assumption}


\section{Two DMHE designs}
A schematic of the partition-based distributed scheme is shown in Figure \ref{schematic:PMHE}. The entire system is decomposed into interacting subsystems. An estimator is developed based on the corresponding subsystem model to handle state estimation for the same subsystem. The estimators for the subsystems are developed based on MHE; they exchange information in a real-time manner to coordinate the state estimates for the subsystems. The state estimates and measurements from other subsystem are shared through a communication network, and all these information are utilized in each local estimator. In this section,
we propose two MHE-based iterative distributed state estimation solutions in which the objective functions of the local estimators are established: 1) DMHE-1 for the case scenario when there are no constraints; 2) DMHE-2 for the case scenario when the states and/or the disturbances of the subsystems are bounded. Both solutions require that the local MHE-based estimators are executed iteratively within every sampling period as new measurements become available. Either of the solutions can be adopted to develop the distributed estimation scheme in Figure \ref{schematic:PMHE}, depending on whether constraints need to be imposed or not.

\begin{figure}
  \centering
  \includegraphics[width=0.65\textwidth]{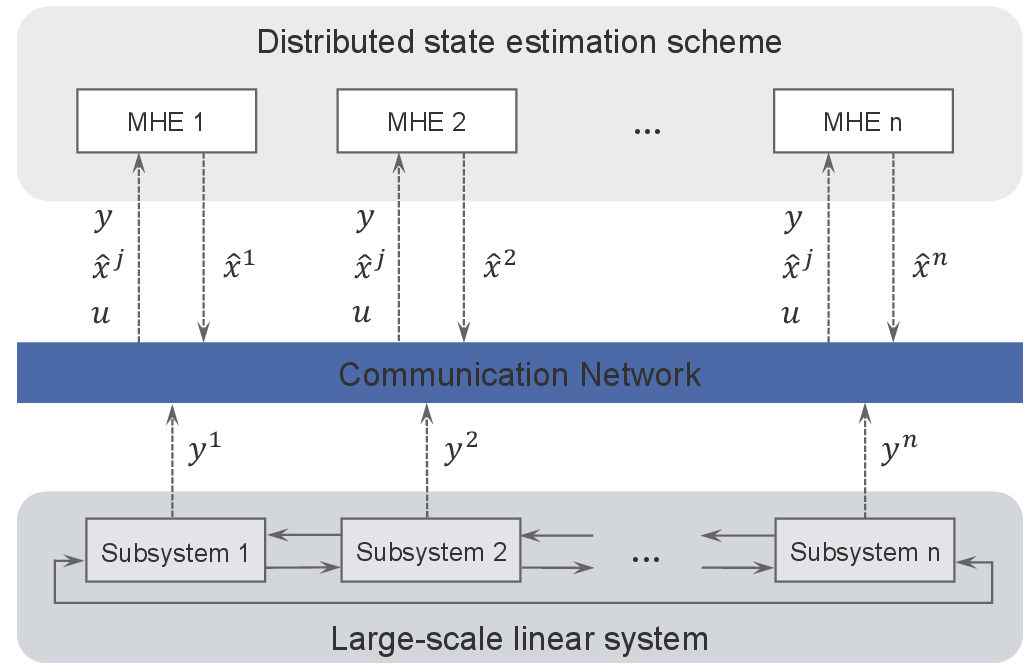}
  \caption{A schematic diagram of the partition-based distributed moving horizon estimation scheme.}\label{schematic:PMHE}
\end{figure}
\subsection{Construction of partition-based objective functions}
In this subsection, we construct partition-based objective functions, which will be used as the local estimators of the two DMHE solutions of this work.
First, the global objective function \eqref{CMHE_obj} of CMHE is partitioned into $\Phi^i_k$, $i\in\mathbb{N}$, such that $\Phi_{k}=\sum_{i=1}^{n}\Phi^{i}_{k}$. Each partitioned function has the following form:
\begin{align}\label{DMHE_obj1}
  \Phi_{k}^{i}= &\frac{1}{2}\Big(\big\|\hat{x}_{k-N|k}^{i}-\bar{x}_{k-N}^{i}\big\|_{P^{-1}_{i}}^{2}+\sum_{j=k-N}^{k-1}\|\hat{w}^{i}_{j}\|^2_{Q_{i}^{-1}}+\sum_{j=k-N}^{k}\|\hat{v}^{i}_{j}\|^2_{R_{i}^{-1}}\Big)
\end{align}
where $\hat{x}^{i}_{k-N|k}$ is an estimate of state $x^{i}_{k-N}$ of the subsystem $i$ calculated at time $k$, and $\bar{x}_{k-N}^{i}$ is a one-step-ahead prediction of the state of the $i$th subsystem, which is computed as:
\begin{equation}\label{priori}
  \bar{x}^{i}_{k-N} = A_{ii}\hat{x}_{k-N-1|k-1}^{i}+\sum_{j\in\mathbb{N}\setminus\{i\}}A_{ij}\hat{x}_{k-N-1|k-1}^{j}+
    B_{ii}u_{k-N-1}^{i}+\sum_{j\in\mathbb{N}\setminus\{i\}}B_{ij}u_{k-N-1}^{j}
\end{equation}
It is noted that \eqref{DMHE_obj1} only takes advantage of the sensor measurements of subsystem $i$. The measured outputs of the interacting subsystems also contain valuable information which may be helpful for reconstructing the state of the current subsystem. Based on this consideration, the sensor measurement information of the interacting subsystems of subsystem $i$ is added to the decomposed function $\Phi_{k}^{i}$ in \eqref{DMHE_obj1} to form an individual objective function for the MHE-based estimator for the $i$th subsystem as follows:
\begin{equation}\label{PMHE_genral}
\bar{\Phi}_{k}^{i}=\Phi_{k}^{i}+\frac{1}{2}\sum_{l\in\mathbb{N}\setminus\{i\}}\sum_{j=k-N}^{k}\|\hat{v}^{l}_{j}\|^2_{R_{l}^{-1}}
\end{equation}

The MHE-based estimators of the distributed estimation scheme are required to be executed for $p_{\mathrm{max}}$ iteration steps within each sampling period. For $p=1,2,\ldots, p_{\max}$, the objective function for the $p$th iteration step of the estimator for subsystem $i$, given the sensor measurements up to sampling instant $k$, is formulated as:
\begin{align}\label{PMHE_iter}
  \bar{\Phi}_{k}^{i,[p]}  =& \frac{1}{2}\Big(\big\|\hat{x}_{k-N|k}^{i,[p]}-\bar{x}_{k-N}^{i}\big\|_{P^{-1}_{i}}^{2}+\sum_{j=k-N}^{k-1}\big\|\hat{w}^{i,[p]}_{j}\big\|^2_{Q_{i}^{-1}}+\sum_{l\in\mathbb{N}}\Big\|\{y^{l}\}_{k-N}^{k}-O^{l}_{[:,i]}\hat{x}_{k-N|k}^{i,[p]}\\
&\quad-\sum_{j\in\mathbb{N}\setminus \{i\}}O^{l}_{[:,j]}\hat{x}_{k-N|k}^{j,[p-1]}-\Lambda^{l}\{u\}_{k-N}^{k-1}-\Gamma^{l}_{[:,i]}\{\hat{w}^{i,[p]}\}_{k-N}^{k-1}-\sum_{j\in\mathbb{N}\setminus \{i\} }\Gamma^{l}_{[:,j]}\{\hat{w}^{j,[p-1]}\}_{k-N}^{k-1}\Big\|^{2}_{\mathbf{R_{l}}^{-1}}\bigg)\nonumber
\end{align}
where $O^{i}$, $\Lambda^{i}$, and $\Gamma^{i}$ are matrices consisting of the rows of matrices $O$, $\Lambda$, and $\Gamma$ with respect to the measurements of subsystem $i$, respectively; $O_{[:,i]}$ is a matrix consisting of the columns of $O$ that are associated with $x^i$; $\Gamma_{[:,i]}$ is a matrix consisting of the columns of $\Gamma$ that are associated with $w^i$; $\mathbf{R_{l}} = \text{diag}\{\underbrace{R_{l}, R_{l}, \ldots, R_{l}}_{N+1}\}$.
At each new sampling instant $k$, the estimator for the $i$th subsystem is initialized with $\hat{x}_{k-N|k}^{i,[0]}=\bar{x}_{k-N}^{i}$. As an analog of the {\em a priori} state prediction in \eqref{priori}, $\bar{x}_{k-N}^{i}$, in \eqref{PMHE_iter} is given by
\begin{equation}\label{pmax}
 \bar{x}^{i}_{k-N} = A_{ii}\hat{x}_{k-N-1|k-1}^{i,[p_{\max}]}+\sum_{j\in\mathbb{N}\setminus\{i\}}A_{ij}\hat{x}_{k-N-1|k-1}^{j,[p_{\max}]}+B_{ii}u_{k-N-1}^{i}+\sum_{j\in\mathbb{N}\setminus\{i\}}B_{ij}u_{k-N-1}^{j}
\end{equation}
for $k>N+1$, where $\hat{x}_{k-N-1|k-1}^{i,[p_{\max}]}$ is the optimal estimate of $x_{k-N-1}^{i}$ provided by the local estimator $i$ in the last iteration step at sampling instant $k-1$.
\subsection{Design of DMHE-2}\label{design_2}
In this subsection, we present the design of the MHE-based estimator for DMHE-2 for the case when hard constraints on system states and disturbances are present.

By leveraging the configured individual objective function in \eqref{PMHE_iter}, the optimal estimate of the state of each subsystem $i$, $i\in\mathbb N$, can be uniquely determined via solving the optimization problem below:
\begin{subequations}\label{eq:10}
\begin{align}
&\quad\quad\quad\quad\quad\quad\quad\quad\quad\min_{\hat{x}^{i,[p]}_{k-N|k},\{\hat{w}^{i,[p]}\}_{k-N}^{k-1}} \bar{\Phi}_{k}^{i,[p]} \label{CPMHE1}\\
&\text{s.t.}\quad  \hat{x}_{t+1|k}^{i}=A_{ii}\hat{x}^{i}_{t|k}+\sum_{j\in\mathbb{N}\setminus\{i\}}A_{ij}\hat{x}_{t|k}^{j}+B_{ii}u_{t}^{i}+\sum_{j\in\mathbb{N}\setminus\{i\}}B_{ij}u_{t}^{j}+\hat{w}^{i}_{t},\label{CPMHE2}\\
& ~\quad\quad\qquad\, t= k-N,\ldots,k-1\nonumber\\[0.1em]
&\quad\quad\qquad\,\hat{x}^{i}\in\mathbb{X}^i,~\hat{w}^{i}\in\mathbb{W}^i,
\end{align}
\end{subequations}
where $\mathbb{X}^i$ and $\mathbb{W}^i$ are convex sets.

By substituting the subsystem model constraint \eqref{CPMHE2} into objective function \eqref{CPMHE1}, the design of the $i$th estimator of DMHE-2 is formulated as follows:
\begin{subequations}\label{DMHE2}
\begin{align}
&\min_{\hat{x}^{i,[p]}_{k-N|k},\{\hat{w}^{i,[p]}\}_{k-N}^{k-1}} \bar{\Phi}_{k}^{i,[p]} \\[0.2em]
&\text{s.t.}~\hat{x}^{i}\in\mathbb{X}^i, ~\hat{w}^{i}\in\mathbb{W}^i\label{constraints1}
\end{align}
where
\begin{align}\label{obj_p}
\bar{\Phi}_{k}^{i,[p]} &=\frac{1}{2}\bigg(\big\|\hat{x}_{k-N|k}^{i,[p]}-\bar{x}_{k-N}^{i}\big\|_{P^{-1}_{i}}^{2}+\big\|\{\hat{w}^{i,[p]}\}_{k-N}^{k-1}\big\|^{2}_{\mathbf{Q}_{i}^{-1}} +\Big\|\{y\}_{k-N}^{k}-O_{[:,i]}\hat{x}_{k-N|k}^{i,[p]}\nonumber\\
&\quad-\sum_{l\in\mathbb{N}\setminus \{i\}}O_{[:,l]}\hat{x}_{k-N|k}^{l,[p-1]}-\Lambda\{u\}_{k-N}^{k-1}-\Gamma_{[:,i]}\{\hat{w}^{i,[p]}\}_{k-N}^{k-1}-\sum_{l\in\mathbb{N}\setminus \{i\} }\Gamma_{[:,l]}\{\hat{w}^{l,[p-1]}\}_{k-N}^{k-1}\Big\|^{2}_{\mathbf{R}^{-1}}\bigg).\nonumber\\
\end{align}
\end{subequations}
In \eqref{obj_p}, $\mathbf{Q}_i = \text{diag}\{\underbrace{Q_i, Q_i, \ldots, Q_i}_{N}\}$; $\mathbf{R} = \text{diag}\{\underbrace{R, R, \ldots, R}_{N+1}\}$; 
The definitions of $O$, $\Lambda$, and $\Gamma$ are given in Section \ref{section:2.3}.

\subsection{Design of DMHE-1}\label{design_1}

When constraints do not need to be imposed (e.g., when disturbances and noise can be approximated by Gaussian distribution thus are considered unbounded), the hard constraints in \eqref{constraints1} can be excluded from the optimization problem for each MHE-based estimator of DMHE-2. In this case, DMHE-2 reduces to the case of unconstrained DMHE, which is called DMHE-1 in this work. Without the constraints in \eqref{constraints1}, we are able to obtain the analytical recursive expression for each local MHE estimator of DMHE-1.  Specifically, we calculate the first-order partial derivatives of the objective function \eqref{obj_p} with respect to the subsystem state estimate $\hat{x}^{i,[p]}_{k-N|k}$ and the estimate of the subsystem disturbances $\{\hat{w}^{i,[p]}\}_{k-N}^{k-1}$, and the following equalities are obtained:
\begin{subequations}\label{first}
  \begin{align}
  \frac{\partial\bar{\Phi}_{k}^{i,[p]}}{\partial\hat{x}^{i,[p]}_{k-N|k}} &=  \Big(P_{i}^{-1}+O_{[:,i]}^{\mathrm{T}}\mathbf{R}^{-1}O_{[:,i]}\Big)\hat{x}^{i,[p]}_{k-N|k}+ \sum_{l\in\mathbb{N}\setminus \{i\} }O_{[:,i]}^{\mathrm{T}}\mathbf{R}^{-1}O_{[:,l]}\hat{x}^{l, [p-1]}_{k-N|k}\nonumber\\
  &\quad+O^{\mathrm{T}}_{[:,i]}\mathbf{R}^{-1}\Gamma_{[:,i]}\{\hat{w}^{i,[p]}\}^{k-1}_{k-N}+\sum_{l\in\mathbb{N}\setminus \{i\} }O^{\mathrm{T}}_{[:,i]}\mathbf{R}^{-1}\Gamma_{[:,l]}\{\hat{w}^{l,[p-1]}\}_{k-N}^{k-1}\label{first_11}\\
   & \quad-P_{i}^{-1}\bar{x}^{i}_{k-N}+O^\mathrm{T}_{[:,i]}\mathbf{R}^{-1}\Lambda\{u\}_{k-N}^{k-1}-O^{\mathrm{T}}_{[:,i]}\mathbf{R}^{-1}\{y\}^{k}_{k-N}\nonumber\\
   &=0\nonumber\\
 \frac{\partial\bar{\Phi}_{k}^{i,[p]}}{\partial\{\hat{w}^{i,[p]}\}_{k-N}^{k-1}}&= \Big(\mathbf{Q}_{i}^{-1}+\Gamma^{\mathrm{T}}_{[:,i]}\mathbf{R}^{-1}\Gamma_{[:,i]}\Big)\{\hat{w}^{i,[p]}\}^{k-1}_{k-N}+ \sum_{l\in\mathbb{N}\setminus \{i\} }\Gamma^{\mathrm{T}}_{[:,i]}\mathbf{R}^{-1}\Gamma_{[:,l]}\{\hat{w}^{l,[p-1]}\}_{k-N}^{k-1}\nonumber \\
 &\quad+\Gamma^{\mathrm{T}}_{[:,i]}\mathbf{R}^{-1}O_{[:,i]}\hat{x}^{i,[p]}_{k-N|k}+\sum_{l\in\mathbb{N}\setminus \{i\} }\Gamma^{\mathrm{T}}_{[:,i]}\mathbf{R}^{-1}O_{[:,l]}\hat{x}^{l,[p-1]}_{k-N|k}\label{first_21}\\
   &\quad+\Gamma^{\mathrm{T}}_{[:,i]}\mathbf{R}^{-1}\Lambda\{u\}_{k-N}^{k-1}-\Gamma^{\mathrm{T}}_{[:,i]}\mathbf{R}^{-1}\{y\}_{k-N}^{k}\nonumber\\
   &= 0\nonumber
\end{align}
\end{subequations}
Based on \eqref{first_11} and \eqref{first_21}, two equations that characterize the optimal estimates of the subsystem state and the sequence of the estimates of the subsystem disturbances are given as follows:
\begin{subequations}\label{derivative}
  \begin{align}
 \Big(P_{i}^{-1}+O_{[:,i]}^{\mathrm{T}}\mathbf{R}^{-1}O_{[:,i]}\Big)\hat{x}^{i,[p]}_{k-N|k} &= -O^{\mathrm{T}}_{[:,i]}\mathbf{R}^{-1}\Gamma_{[:,i]}\{\hat{w}^{i,[p]}\}^{k-1}_{k-N} -\sum_{l\in\mathbb{N}\setminus \{i\} }O_{[:,i]}^{\mathrm{T}}\mathbf{R}^{-1}O_{[:,l]}\hat{x}^{l, [p-1]}_{k-N|k} \nonumber\\
  & \quad    -\sum_{l\in\mathbb{N}\setminus \{i\} }O^{\mathrm{T}}_{[:,i]}\mathbf{R}^{-1}\Gamma_{[:,l]}\{\hat{w}^{l,[p-1]}\}_{k-N}^{k-1}+P_{i}^{-1}\bar{x}^{i}_{k-N}\label{derivative_1}\\
   & \quad-O^T_{[:,i]}\mathbf{R}^{-1}\Lambda\{u\}_{k-N}^{k-1}+O^{\mathrm{T}}_{[:,i]}\mathbf{R}^{-1}\{y\}^{k}_{k-N}\nonumber\\
\Big(\mathbf{Q}_{i}^{-1}+\Gamma^{\mathrm{T}}_{[:,i]}\mathbf{R}^{-1}\Gamma_{[:,i]}\Big)\{\hat{w}^{i,[p]}\}^{k-1}_{k-N} &=  -\Gamma^{\mathrm{T}}_{[:,i]}\mathbf{R}^{-1}O_{[:,i]}\hat{x}^{i,[p]}_{k-N|k}  - \sum_{l\in\mathbb{N}\setminus \{i\} }\Gamma^{\mathrm{T}}_{[:,i]}\mathbf{R}^{-1}\Gamma_{[:,l]}\{\hat{w}^{l,[p-1]}\}_{k-N}^{k-1}\nonumber \\
 &\quad-\sum_{l\in\mathbb{N}\setminus \{i\} }\Gamma^{\mathrm{T}}_{[:,i]}\mathbf{R}^{-1}O_{[:,l]}\hat{x}^{l,[p-1]}_{k-N|k}\label{derivative_2}\\
   &\quad-\Gamma^{\mathrm{T}}_{[:,i]}\mathbf{R}^{-1}\Lambda\{u\}_{k-N}^{k-1}+\Gamma^{\mathrm{T}}_{[:,i]}\mathbf{R}^{-1}\{y\}_{k-N}^{k}\nonumber
\end{align}
\end{subequations}

Theorem \ref{theorem:0} below presents the solution to \eqref{derivative}; this solution is the recursive expression for each estimator of DMHE-1.
\begin{theorem}\label{theorem:0}
In the $p$th iteration step of sampling instant $k$, given $\{y\}_{k-N}^{k}$, $\{u\}_{k-N}^{k-1}$, $\bar{x}_{k-N}^{i}$, and $\hat{x}_{k-N|k}^{l,[p-1]}$, $l\in\mathbb{N}\setminus \{i\}$, the $i$th estimator of DMHE-1 computes the optimal estimates of the subsystem state and disturbances as below:
\begin{subequations}\label{solution}
\begin{align}
  \hat{x}^{i,[p]}_{k-N|k} &= \Big(P_{i}^{-1}+O_{[:,i]}^{\mathrm{T}}\big(\mathbf{R}+\Gamma_{[:,i]}\mathbf{Q}_{i}\Gamma_{[:,i]}^{\mathrm{T}}\big)^{-1}O_{[:,i]}\Big)^{-1}\\
  &\quad\times\Big(c_{1}-O^{\mathrm{T}}_{[:,i]}\mathbf{R}^{-1}\Gamma_{[:,i]}\big(\mathbf{Q}_{i}^{-1}+\Gamma^{\mathrm{T}}_{[:,i]}\mathbf{R}^{-1}\Gamma_{[:,i]}\big)^{-1}c_{2}\Big)\nonumber\\
  \{\hat{w}^{i,[p]}\}^{k-1}_{k-N}&= \Big(\mathbf{Q}_{i}^{-1}+\Gamma_{[:,i]}^{\mathrm{T}}\big(\mathbf{R}+O_{[:,i]}P_{i}O_{[:,i]}^{\mathrm{T}}\big)^{-1}\Gamma_{[:,i]}\Big)^{-1}\\
  &\quad\times\Big(-\Gamma^{\mathrm{T}}_{[:,i]}\mathbf{R}^{-1}O_{[:,i]}\big(P_{i}^{-1}+O_{[:,i]}^{\mathrm{T}}\mathbf{R}^{-1}O_{[:,i]}\big)^{-1}c_{1}+c_{2}\Big)\nonumber
\end{align}
{\text{where}}
\begin{align*}
  c_{1} &=-\sum_{l\in\mathbb{N}\setminus \{i\} }O_{[:,i]}^{\mathrm{T}}\mathbf{R}^{-1}O_{[:,l]}\hat{x}^{l, [p-1]}_{k-N|k}-\sum_{l\in\mathbb{N}\setminus \{i\} }O^{\mathrm{T}}_{[:,i]}\mathbf{R}^{-1}\Gamma_{[:,l]}\{\hat{w}^{l,[p-1]}\}_{k-N}^{k-1}\\
  & \quad    +P_{i}^{-1}\bar{x}^{i}_{k-N}-O^T_{[:,i]}\mathbf{R}^{-1}(\Lambda\{u\}_{k-N}^{k-1}-\{y\}^{k}_{k-N})\\
   c_{2} &=-\sum_{l\in\mathbb{N}\setminus \{i\} }\Gamma^{\mathrm{T}}_{[:,i]}\mathbf{R}^{-1}O_{[:,l]}\hat{x}^{l,[p-1]}_{k-N|k} - \sum_{l\in\mathbb{N}\setminus \{i\} }\Gamma^{\mathrm{T}}_{[:,i]}\mathbf{R}^{-1}\Gamma_{[:,l]}\{\hat{w}^{l,[p-1]}\}_{k-N}^{k-1} \\
 &\quad-\Gamma^{\mathrm{T}}_{[:,i]}\mathbf{R}^{-1}(\Lambda\{u\}_{k-N}^{k-1}-\{y\}_{k-N}^{k})
\end{align*}
\end{subequations}
\end{theorem}
Lemmas \ref{inversion} and \ref{lemma:2} below will be used for proving Theorem~\ref{theorem:0}.
\begin{lemma}(Matrix Inversion Lemma \cite{henderson1981deriving})\label{inversion}
For nonsingular matrices $A$ and $D$, it holds that
\begin{equation}
(A-BD^{-1}C)^{-1} = A^{-1}+A^{-1}B(D-CA^{-1}B)^{-1}CA^{-1}
\end{equation}
\end{lemma}
\begin{lemma}\cite{zhang2006schur}\label{lemma:2}
Let $M$ be a block matrix:
\begin{equation*}
M=\left[
\begin{array}{cc}
    A & B  \\
    C & D
\end{array}\right]
\end{equation*}
where $A$, $B$, $C$, and $D$ are matrices of compatible dimensions. If $A$ and $D$ are invertible, then:
\begin{equation*}
M^{-1}=
\left[\begin{array}{cc}
  (A-BD^{-1}C)^{-1}   & -(A-BD^{-1}C)^{-1}BD^{-1} \\
  -(D-CA^{-1}B)^{-1}CA^{-1}  & (D-CA^{-1}B)^{-1}
\end{array}\right]
\end{equation*}
\end{lemma}
\begin{proof}
Let
\begin{align}\label{m_definition}
  &m_{11} = P_{i}^{-1}+O_{[:,i]}^{\mathrm{T}}\mathbf{R}^{-1}O_{[:,i]},\quad \quad  m_{12}=O^{\mathrm{T}}_{[:,i]}\mathbf{R}^{-1}\Gamma_{[:,i]}, \\
  &m_{21}=\Gamma^{\mathrm{T}}_{[:,i]}\mathbf{R}^{-1}O_{[:,i]}, \quad \quad \quad \quad \quad  m_{22}=\mathbf{Q}_{i}^{-1}+\Gamma^{\mathrm{T}}_{[:,i]}\mathbf{R}^{-1}\Gamma_{[:,i]}.\nonumber
\end{align}
Based on \eqref{solution}, Lemma \ref{inversion} and Lemma \ref{lemma:2}, the estimates of the state and the disturbances of the $i$th subsystem in the $p$th iteration step can be calculated as:
\begin{equation}\label{eq:15}
  \left[\begin{array}{c}
          \hat{x}^{i,[p]}_{k-N|k} \\
          \{\hat{w}^{i,[p]}\}^{k-1}_{k-N}
        \end{array}\right]
        =
        \left[\begin{array}{cc}
          m_{11} & m_{12} \\
          m_{21} & m_{22}
        \end{array}\right]^{-1}
         \left[\begin{array}{c}
          c_{1} \\
          c_{2}
        \end{array}\right]
\end{equation}
Based on Lemma \ref{lemma:2}, we can further derive that
\begin{align}\label{eq:16}
  &\left[\begin{array}{c}
          \hat{x}^{i,[p]}_{k-N|k} \\
          \{\hat{w}^{i,[p]}\}^{k-1}_{k-N}
        \end{array}\right]\\
        =&\left[\begin{array}{cc}
          (m_{11}-m_{12}m_{22}^{-1}m_{21})^{-1} & -(m_{11}-m_{12}m_{22}^{-1}m_{21})^{-1}m_{12}m_{22}^{-1} \\
         -(m_{22}-m_{21}m_{11}^{-1}m_{12}\big)^{-1}m_{21}m_{11}^{-1} & (m_{22}-m_{21}m_{11}^{-1}m_{12}\big)^{-1}
        \end{array}\right]
        \left[\begin{array}{c}
          c_{1} \\
          c_{2}
        \end{array}\right]\nonumber
\end{align}
Further, the estimates of the state and the estimates of the disturbances can be presented separately:
\begin{subequations}\label{eq:17}
\begin{align}
  \hat{x}^{i,[p]}_{k-N|k} & = \big(m_{11}-m_{12}m_{22}^{-1}m_{21}\big)^{-1}\big(c_{1} -m_{12}m_{22}^{-1}c_{2}\big) \label{eq:17a}\\
  \{\hat{w}^{i,[p]}\}^{k-1}_{k-N} & = \big(m_{22}-m_{21}m_{11}^{-1}m_{12}\big)^{-1}\big(-m_{21}m_{11}^{-1}c_{1}+c_{2}\big)\label{eq:17b}
\end{align}
\end{subequations}
By applying Lemma \ref{inversion}, the following transformation is conducted for \eqref{eq:17a}:
\begin{align}\label{form_17a}
  &\quad\big(m_{11}-m_{12}m_{22}^{-1}m_{21}\big)^{-1} \nonumber \\
   &=\Big(P_{i}^{-1}+O_{[:,i]}^{\mathrm{T}}\mathbf{R}^{-1}O_{[:,i]}-O^{\mathrm{T}}_{[:,i]}\mathbf{R}^{-1}\Gamma_{[:,i]}\big(\mathbf{Q}_{i}^{-1}+\Gamma^{\mathrm{T}}_{[:,i]}\mathbf{R}^{-1}\Gamma_{[:,i]}\big)^{-1}\Gamma^{\mathrm{T}}_{[:,i]}\mathbf{R}^{-1}O_{[:,i]}\Big)^{-1}  \\
  & = \Big(P_{i}^{-1}+O_{[:,i]}^{\mathrm{T}}\Big(\mathbf{R}^{-1}-\mathbf{R}^{-1}\Gamma_{[:,i]}\big(\mathbf{Q}_{i}^{-1}+\Gamma^{\mathrm{T}}_{[:,i]}\mathbf{R}^{-1}\Gamma_{[:,i]}\big)^{-1}\Gamma^{\mathrm{T}}_{[:,i]}\mathbf{R}^{-1}\Big)O_{[:,i]}\Big)^{-1}\nonumber  \\
  &=\Big(P_{i}^{-1}+O_{[:,i]}^{\mathrm{T}}\big(\mathbf{R}+\Gamma_{[:,i]}\mathbf{Q}_{i}\Gamma^{\mathrm{T}}_{[:,i]}\big)^{-1}O_{[:,i]}\Big)^{-1} \nonumber
\end{align}
Similarly, for \eqref{eq:17b}, it holds that:
\begin{align}\label{form_17b}
  & \quad (m_{22}-m_{21}m_{11}^{-1}m_{12})^{-1} \nonumber \\
   &=\Big(\mathbf{Q}_{i}^{-1}+\Gamma^{\mathrm{T}}_{[:,i]}\mathbf{R}^{-1}\Gamma_{[:,i]}-\Gamma^{\mathrm{T}}_{[:,i]}\mathbf{R}^{-1}O_{[:,i]}\big( P_{i}^{-1}+O_{[:,i]}^{\mathrm{T}}\mathbf{R}^{-1}O_{[:,i]}\big)^{-1}O^{\mathrm{T}}_{[:,i]}\mathbf{R}^{-1}\Gamma_{[:,i]}\Big)^{-1} \\
   & =\Big(\mathbf{Q}_{i}^{-1}+\Gamma^{\mathrm{T}}_{[:,i]}\Big(\mathbf{R}^{-1}-\mathbf{R}^{-1}O_{[:,i]}\big( P_{i}^{-1}+O_{[:,i]}^{\mathrm{T}}\mathbf{R}^{-1}O_{[:,i]}\big)^{-1}O^{\mathrm{T}}_{[:,i]}\mathbf{R}^{-1}\Big)\Gamma_{[:,i]}\Big)^{-1}\nonumber\\
   &=\Big(\mathbf{Q}_{i}^{-1}+\Gamma^{\mathrm{T}}_{[:,i]}\big(\mathbf{R}+O_{[:,i]} P_{i}O_{[:,i]}^{\mathrm{T}}\big)^{-1}\Gamma_{[:,i]}\Big)^{-1}\nonumber
\end{align}
By substituting \eqref{m_definition}, \eqref{form_17a}, and \eqref{form_17b} into \eqref{eq:17}, \eqref{solution} is proven. $\square$
\end{proof}
\begin{rmk}
Iterative distributed state estimation methods have the potential to provide estimates that converge to the centralized counterpart. When faster convergence of the estimation error and/or more accurate estimates are needed, iterative distributed designs can be more favorable as compared to non-iterative algorithms (for which the local estimators are only executed once at each sampling instant) \cite{schneider2015convergence, schneider2017solution}.
Based on this consideration, as is similar to DMHE-2, DMHE-1 is also designed in a way such that the local estimators within the distributed framework are executed iteratively for multiple times within each sampling period.
As the unconstrained local estimators of DMHE-1 are executed for multiple times at each sampling instant, the estimates of the subsystems given by the local estimators are updated iteratively at each sampling instant, which are characterized by the analytical recursive expression in \eqref{solution}.
\end{rmk}

\subsection{Implementation of DMHE algorithm}
We present Algorithm \ref{alg1} as the implementation algorithm that can be followed to execute the proposed DMHE designs to generate optimal state estimate $\hat{x}_{k-N|k}^{i}$ for each subsystem $i$, $i\in\mathbb N$, of the entire system in \eqref{model}. This algorithm describes the iterative executions of the local estimators for both DMHE-1 and DMHE-2.
\begin{algorithm}[t]
\caption{Execution of the proposed distributed estimation method (DMHE-1 or DMHE-2)}\vspace{3mm}
\label{alg1}
At time $k\geq N+1$, MHE $i$, $i\in\mathbb{N}$, conducts the following steps:
\begin{enumerate}
    \item[1.] Receive measured outputs $\{y\}_{k-N}^{k}$, inputs $\{u\}_{k-N}^{k-1}$, and optimal estimate $\hat{x}^{j}_{k-N-1|k-1}$ of $j$th subsystem, $j\in\mathbb{N}\setminus\{i\}$, at the last time step $k-1$ from the estimators of interacting subsystems.
    \item[2.] Compute the open-loop prediction $\bar{x}_{k-N}^{i}$ via \eqref{pmax}.
    \item[3.] Initialize itself with $\hat{x}_{k-N|k}^{i,[0]} = \bar{x}_{k-N}^{i}$.
    \item[4.]Set $p=1$.
    \begin{enumerate}
        \item [4.1.] Receive the most recent estimates $\hat{x}_{k-N|k}^{j,[p-1]}$ and $\{\hat{w}^{j,[p-1]}\}_{k-N}^{k-1}$ from estimator $j$, $j\in\mathbb{N}\setminus\{i\}$.
        \item[4.2.] Calculate $\hat{x}_{k-N|k}^{i,[p]}$ and $\{\hat{w}^{i,[p]}\}_{k-N}^{k-1}$ by solving
    \eqref{solution} for DMHE-1, or solving \eqref{DMHE2} if DMHE-2 is adopted.
        \item[4.3.]
        {\emph{\textbf{If}}} $p=p_{\mathrm{max}}$, do:
    \begin{itemize}
        \item The estimate $\hat{x}_{k-N|k}^{i,[p_{\mathrm{max}}]}$  is treated as an optimal estimate $\hat{x}_{k-N|k}^{i}$ for subsystem $i$, $i\in\mathbb{N}$. Go to step 5.
    \end{itemize}
     {\emph{\textbf{Else}}}, do:
     \begin{itemize}
         \item Set $p=p+1$. Go to step 4.1.
     \end{itemize}
    \end{enumerate}
     \item[5.]\label{final} Set $k=k+1$. Go to step 1.
     \end{enumerate}

\end{algorithm}

\section{Stability analysis}
This section presents sufficient conditions for the convergence of the state estimate and stability of the estimation error dynamics of the proposed partition-based distributed estimation method for both the unconstrained case (DMHE-1) and the constrained case (DMHE-2).

\subsection{Convergence of the unconstrained iterative DMHE (DMHE-1)}
The MHE-based estimators of the distributed estimation mechanism (developed based on either DMHE-1 or DMHE-2) are required to be executed iteratively at each new sampling instant as new sensor measurements are made available. It is necessary to ensure the estimates of the MHE-based estimators will converge as the estimators are executed iteratively within each sampling period. In this subsection, we present sufficient conditions on the convergence of the subsystem state estimate as DMHE-1 is executed iteratively.

\begin{theorem}\label{theorem1}
In the $p$th iteration step at sampling instant $k$, the solution to DMHE-1 for the entire system \eqref{model} can be described by the following form:
\begin{equation}\label{DMHE_no}
\left[\begin{array}{c}
                 \hat{x}_{k-N|k}^{[p]}  \\
                  \{\hat{w}^{[p]}\}_{k-N}^{k-1}
               \end{array}\right]= - M_{d}^{-1}M_{r}\left[\begin{array}{c}
                 \hat{x}_{k-N|k}^{[p-1]}  \\
                  \{\hat{w}^{[p-1]}\}_{k-N}^{k-1}
               \end{array}\right]+M_{k}.
\end{equation}
In addition, the estimates $\hat{x}_{k-N|k}^{[p]}$ and $\{\hat{w}^{[p]}\}_{k-N}^{k-1}$ converge asymptotically to the unique solution to CMHE in \eqref{CMHE_opt}, if all eigenvalues of the matrix $M_{d}^{-1}M_{r}$ are inside the unit circle, i.e., the spectral radius
\begin{equation}\label{convergence_no}
\rho(M_{d}^{-1}M_{r})<1
\end{equation}
where
\begin{align*}
  M_r&=\left[\begin{matrix}
                  \big(O^{\mathrm{T}}\mathbf{R}^{-1}O\big)_{r} & \big(O^{\mathrm{T}}\mathbf{R}^{-1}\Gamma\big)_{r}  \\
                   \big(\Gamma^{\mathrm{T}}\mathbf{R}^{-1}O\big)_{r}& \big(\Gamma^{\mathrm{T}}\mathbf{R}^{-1}\Gamma\big)_{r}
                \end{matrix}\right],
  M_d  =\left[\begin{matrix}
                  \big(P^{-1}+O^{\mathrm{T}}\mathbf{R}^{-1}O\big)_{d} & \big(O^{\mathrm{T}}\mathbf{R}^{-1}\Gamma\big)_{d}  \\
                   \big(\Gamma^{\mathrm{T}}\mathbf{R}^{-1}O\big)_{d}& \big(\mathbf{Q}^{-1}+\Gamma^{\mathrm{T}}\mathbf{R}^{-1}\Gamma\big)_{d}
                \end{matrix}\right],\\
  M_{k} & =\left[\begin{array}{c}
         P^{-1}\bar{x}_{k-N}+O^{\mathrm{T}}\mathbf{R}^{-1}\Big(\{y\}^{k}_{k-N}-\Lambda\{u\}_{k-N}^{k-1}\Big) \\
         \Gamma^{\mathrm{T}}\mathbf{R}^{-1}\Big(\{y\}^{k}_{k-N}-\Lambda\{u\}_{k-N}^{k-1}\Big)
       \end{array}\right],\\
  \mathbf{Q}&= \mathrm{diag}\{\mathbf{Q}_1, \mathbf{Q}_2, \ldots, \mathbf{Q}_n\},~
   P= \mathrm{diag}\{P_1, P_2, \ldots, P_n\}.
\end{align*}
\end{theorem}
\begin{proof}

The following equality, which will be used to derive \eqref{DMHE_no}, holds:
\begin{align}\label{transformation}
    & \left[ \begin{array}{c}
                                                O_{[:,1]}^{{\rm{\mathrm{T}}}}\mathbf{R}^{-1}O_{[:,1]}\hat{x}^{1}_{k-N|k}+\sum_{l\in\mathbb{N}\setminus\{1\}}O_{[:,1]}^{{\rm{\mathrm{T}}}}\mathbf{R}^{-1}O_{[:,l]}\hat{x}^{l}_{k-N|k} \\
                                                O_{[:,2]}^{{\rm{\mathrm{T}}}}\mathbf{R}^{-1}O_{[:,2]}\hat{x}^{2}_{k-N|k}+\sum_{l\in\mathbb{N}\setminus\{2\}}O_{[:,2]}^{{\rm{\mathrm{T}}}}\mathbf{R}^{-1}O_{[:,l]}\hat{x}^{l}_{k-N|k} \\
                                                \vdots\\
                                                O_{[:,n]}^{{\rm{\mathrm{T}}}}\mathbf{R}^{-1}O_{[:,n]}\hat{x}^{n}_{k-N|k}+\sum_{l\in\mathbb{N}\setminus\{n\}}O_{[:,n]}^{{\rm{\mathrm{T}}}}\mathbf{R}^{-1}O_{[:,l]}\hat{x}^{l}_{k-N|k}
                                              \end{array}\right]\\
    &= \left[ \begin{array}{c}
                                                O_{[:,1]}^{{\rm{\mathrm{T}}}}\mathbf{R}^{-1}\sum_{l=1}^{n}O_{[:,l]}\hat{x}^{l}_{k-N|k} \\
                                                O_{[:,2]}^{{\rm{\mathrm{T}}}}\mathbf{R}^{-1}\sum_{l=1}^{n}O_{[:,l]}\hat{x}^{l}_{k-N|k} \\
                                                \vdots\\
                                                O_{[:,n]}^{{\rm{\mathrm{T}}}}\mathbf{R}^{-1}\sum_{l=1}^{n}O_{[:,l]}\hat{x}^{l}_{k-N|k}
                                              \end{array}\right]\nonumber
    =\left[ \begin{array}{c}
                                                O_{[:,1]}^{{\rm{\mathrm{T}}}}\mathbf{R}^{-1}O\hat{x}_{k-N|k} \\
                                                O_{[:,2]}^{{\rm{\mathrm{T}}}}\mathbf{R}^{-1}O\hat{x}_{k-N|k} \\
                                                \vdots\\
                                                O_{[:,n]}^{{\rm{\mathrm{T}}}}\mathbf{R}^{-1}O\hat{x}_{k-N|k}
                                              \end{array}\right]\\
    &=O^{{\rm{\mathrm{T}}}}\mathbf{R}^{-1}O\hat{x}_{k-N|k}\nonumber
\end{align}
where $O=[O_{[:,1]},O_{[:,2]},\ldots,O_{[:,n]}]$ and $\hat{x}_{k-N|k}=\mathrm{col}(\hat{x}^{1}_{k-N|k},\hat{x}^{2}_{k-N|k}, \ldots,\hat{x}^{n}_{k-N|k})$.

Based on \eqref{transformation}, concatenating $O_{[:,i]}^{\mathrm{T}}\mathbf{R}^{-1}O_{[:,i]}\hat{x}^{i,[p]}_{k-N|k}$, $\sum_{l\in\mathbb{N}\setminus \{i\} }O_{[:,i]}^{\mathrm{T}}\mathbf{R}^{-1}O_{[:,l]}\hat{x}^{l, [p-1]}_{k-N|k}$ in \eqref{derivative_1} and \eqref{derivative_2} for all $i\in \mathbb{N}$,  can yield $(O^{\mathrm{T}}\mathbf{R}^{-1}O)_{d}\hat{x}^{[p]}_{k-N|k}$ and $(O^{\mathrm{T}}\mathbf{R}^{-1}O)_{r}\hat{x}^{[p-1]}_{k-N|k}$, respectively, where $(O^{\mathrm{T}}\mathbf{R}^{-1}O)_{d}=\mathrm{diag}\big\{O_{[:,1]}^{\mathrm{T}}\mathbf{R}^{-1}O_{[:,1]}, \ldots, O_{[:,n]}^{\mathrm{T}}\mathbf{R}^{-1}O_{[:,n]}\big\}$ and $(O^{\mathrm{T}}\mathbf{R}^{-1}O)_{d}+(O^{\mathrm{T}}\mathbf{R}^{-1}O)_{r}=O^{\mathrm{T}}\mathbf{R}^{-1}O$.
By concatenating the other terms in \eqref{derivative_1} and \eqref{derivative_2} for all $i\in \mathbb{N}$ in a similar manner, we can obtain:
\begin{subequations}\label{eq:25}
  \begin{align}
\big(P^{-1}+O^{\mathrm{T}}\mathbf{R}^{-1}O\big)_{d}\hat{x}_{k-N|k}^{[p]}&=-\big(O^{\mathrm{T}}\mathbf{R}^{-1}\Gamma\big)_{d}\{\hat{w}^{[p]}\}_{k-N}^{k-1} -  \big(O^{\mathrm{T}}\mathbf{R}^{-1}O\big)_{r}\hat{x}_{k-N|k}^{[p-1]}\nonumber \\
 &\quad -\big(O^{\mathrm{T}}\mathbf{R}^{-1}\Gamma\big)_{r} \{\hat{w}^{[p-1]}\}_{k-N}^{k-1} +P^{-1}\bar{x}_{k-N}+O^{\mathrm{T}}\mathbf{R}^{-1}\{y\}_{k-N}^{k}\label{sort_out_1}\\
      & \quad-O^{\mathrm{T}}\mathbf{R}^{-1}\Lambda\{u\}_{k-N}^{k-1} \nonumber\\
   \big(\mathbf{Q}^{-1}+\Gamma^{\mathrm{T}}\mathbf{R}^{-1}\Gamma\big)_d\{\hat{w}^{[p]}\}_{k-N}^{k-1}   &=  -(\Gamma^{\mathrm{T}}\mathbf{R}^{-1}O)_d\hat{x}_{k-N|k}^{[p]}- (\Gamma^{\mathrm{T}}\mathbf{R}^{-1}\Gamma)_r\{\hat{w}^{[p-1]}\}_{k-N}^{k-1}\label{sort_out_2}\\
&\quad-(\Gamma^{\mathrm{T}}\mathbf{R}^{-1}O)_r\hat{x}_{k-N|k}^{[p-1]}-\Gamma^{\mathrm{T}}\mathbf{R}^{-1}\Lambda\{u\}_{k-N}^{k-1}+\Gamma^{\mathrm{T}}\mathbf{R}^{-1}\{y\}_{k-N}^{k}\nonumber
\end{align}
\end{subequations}
By presenting \eqref{sort_out_1} and \eqref{sort_out_2} using a more compact form, the recursion in \eqref{DMHE_no} for DMHE-1 is proven.
In addition, when $\rho(M_{d}^{-1}M_{r})<1$, \eqref{DMHE_no} is asymptotically stable, which ensures the convergence of the estimates to the corresponding centralized MHE in (\ref{CMHE_opt}). $\square$
\end{proof}

\subsection{Stability of the unconstrained DMHE (DMHE-1)}
In this subsection, we prove the stability of the estimation error dynamics for DMHE-1 under the deterministic setting.
\begin{theorem}\label{theorem:2}
If Assumption \ref{assume2} holds, if \eqref{convergence_no} is satisfied, and if the MHE-based estimator of DMHE-1 are executed iteratively for infinite steps at each sampling instant $k$, then the estimation error $e_{k-N}=x_{k-N}-\hat{x}_{k-N|k}$ for the entire process (\ref{model}) given by DMHE-1 is described by:
\begin{equation*}
e_{k-N} = (P^{-1}+O^{\mathrm{T}}(\mathbf{R}+\Gamma \mathbf{Q}\Gamma^{\mathrm{T}})^{-1}O)^{-1}P^{-1}Ae_{k-N-1}.
\end{equation*}
In addition, if
\begin{equation}\label{stability_no}
\rho\left((P^{-1}+O^{\mathrm{T}}(\mathbf{R}+\Gamma \mathbf{Q}\Gamma^{\mathrm{T}})^{-1}O)^{-1}P^{-1}A\right)<1
\end{equation}
holds, then the estimation error provided by DMHE-1 asymptotically converges to zero as $k\rightarrow\infty$.
\end{theorem}

\begin{proof}
If the MHE-based estimators of DMHE-1 are executed for infinite iteration steps at every sampling instant $k$, then \eqref{DMHE_no} can be rewritten as
\begin{equation}\label{DMHE_no_infinity}
         \left[\begin{matrix}
                  P^{-1}+O^{\mathrm{T}}\mathbf{R}^{-1}O& O^{\mathrm{T}}\mathbf{R}^{-1}\Gamma  \\
                   \Gamma^{\mathrm{T}}\mathbf{R}^{-1}O& \mathbf{Q}^{-1}+\Gamma^{\mathrm{T}}\mathbf{R}^{-1}\Gamma
                \end{matrix}\right]
   \left[\begin{array}{c}
  \hat{x}_{k-N|k} \\
  \big\{\hat{w}\big\}^{k-1}_{k-N}
\end{array}\right]
  =  \left[\begin{array}{c}
P^{-1}\bar{x}_{k-N}+O^{\mathrm{T}}\mathbf{R}^{-1}\big(\{y\}^{k}_{k-N}-\Lambda\{u\}_{k-N}^{k-1}\big) \\
         \Gamma^{\mathrm{T}}\mathbf{R}^{-1}\big(\{y\}_{k-N}^{k}-\Lambda\{u\}_{k-N}^{k-1}\big)
       \end{array}\right]
  \end{equation}
$P^{-1}+O^{\mathrm{T}}\mathbf{R}^{-1}O$ and $\mathbf{Q}^{-1}+\Gamma^{\mathrm{T}}\mathbf{R}^{-1}\Gamma$ are positive definite and are invertible.
Moreover, according to \eqref{DMHE_no_infinity}, Lemma \ref{inversion}, and Lemma \ref{lemma:2},
 the optimal estimates of the state and disturbances can be computed following:
  \begin{equation}\label{eq:25}
   \left[\begin{array}{c}
  \hat{x}_{k-N|k} \\
  \{\hat{w}\}^{k-1}_{k-N}
  \end{array}\right]
  =  \left[\begin{array}{cc}
             H_{11} & H_{12} \\
             H_{21} & H_{22}
           \end{array}
  \right]
 \left[\begin{array}{c}
P^{-1}\bar{x}_{k-N}+O^{\mathrm{T}}\mathbf{R}^{-1}\big(\{y\}^{k}_{k-N}-\Lambda\{u\}_{k-N}^{k-1}\big) \\
         \Gamma^{\mathrm{T}}\mathbf{R}^{-1}\big(\{y\}_{k-N}^{k}-\Lambda\{u\}_{k-N}^{k-1}\big)
       \end{array}\right]
  \end{equation}
where
\begin{align}\label{H}
  H_{11} &=\big(P^{-1}+O^{\mathrm{T}}\mathbf{R}^{-1}O-O^{\mathrm{T}}\mathbf{R}^{-1}\Gamma\big(\mathbf{Q}^{-1}+\Gamma^{\mathrm{T}}\mathbf{R}^{-1}\Gamma\big)^{-1}\Gamma^{\mathrm{T}}\mathbf{R}^{-1}O\big)^{-1}\nonumber\\
   & = \big(P^{-1}+O^{\mathrm{T}}\big(\mathbf{R}^{-1}-\mathbf{R}^{-1}\Gamma \big(\mathbf{Q}^{-1}+\Gamma^{\mathrm{T}}\mathbf{R}^{-1}\Gamma\big)^{-1}\Gamma^{\mathrm{T}}\mathbf{R}^{-1}\big)O)^{-1}\\
   &= \big(P^{-1}+O^{\mathrm{T}}\big(\mathbf{R}+\Gamma \mathbf{Q}\Gamma^{\mathrm{T}}\big)^{-1}O\big)^{-1}\nonumber
\end{align}
and
\begin{equation*}
  H_{12}=-H_{11}O^{\mathrm{T}}R\Gamma(Q^{-1}+\Gamma^{\mathrm{T}}R^{-1}\Gamma)^{-1}
\end{equation*}
From (\ref{eq:25}),
\begin{equation}\label{DMHE_no_infinity2}
  \hat{x}_{k-N|k}=H_{11}P^{-1}\bar{x}_{k-N}+\big(H_{11}O^{\mathrm{T}}\mathbf{R}^{-1}+H_{12}\Gamma^{\mathrm{T}}\mathbf{R}^{-1}\big)\Big(\{y\}_{k-N}^{k}-\Lambda\{u\}_{k-N}^{k-1}\Big)
\end{equation}
In the absence of system disturbances and measurement noise, the estimation error is calculated as
\begin{align*}
  e_{k-N} & = \hat{x}_{k-N|k} -x_{k-N} \\
   & = H_{11}P^{-1}\bar{x}_{k-N}+(H_{11}O^{\mathrm{T}}\mathbf{R}^{-1}+H_{12}\Gamma^{\mathrm{T}}\mathbf{R}^{-1})Ox_{k-N}-x_{k-N}
\end{align*}
Due to the fact that $\bar{x}_{k-N}=A\hat{x}_{k-N-1|k-1}+Bu_{k-N-1}$, based on model \eqref{model1}, it yields
\begin{align}\label{error1}
  e_{k-N} &= H_{11}P^{-1}A\hat{x}_{k-N-1|k-1}+ H_{11}P^{-1}Bu_{k-N-1}\nonumber\\
  &\quad+\big(H_{11}O^{\mathrm{T}}\mathbf{R}^{-1}O+H_{12}\Gamma^{\mathrm{T}}\mathbf{R}^{-1}O-I\big)x_{k-N}\nonumber\\
  &=H_{11}P^{-1}Ae_{k-N-1}+H_{11}P^{-1}Ax_{k-N-1}+ H_{11}P^{-1}Bu_{k-N-1}\\
  &\quad+\big(H_{11}O^{\mathrm{T}}\mathbf{R}^{-1}O+H_{12}\Gamma^{\mathrm{T}}\mathbf{R}^{-1}O-I\big)(Ax_{k-N-1}+Bu_{k-N-1})\nonumber\\
  &=\big(H_{11}O^{\mathrm{T}}\mathbf{R}^{-1}O+H_{12}\Gamma^{\mathrm{T}}\mathbf{R}^{-1}O+H_{11}P^{-1}-I\big)(Ax_{k-N-1}+Bu_{k-N-1})\nonumber\\
  &\quad+H_{11}P^{-1}Ae_{k-N-1}\nonumber
\end{align}
Based on \eqref{H} and the Lemma \ref{inversion}, we can further obtain that
\begin{align}\label{matrix_inverse_apply}
    & \quad H_{11}O^{\mathrm{T}}\mathbf{R}^{-1}O+H_{12}\Gamma^{\mathrm{T}}\mathbf{R}^{-1}O+H_{11}P^{-1} \nonumber\\
  &= H_{11}\Big(P^{-1}+O^{\mathrm{T}}\mathbf{R}^{-1}O-O^{\mathrm{T}}\mathbf{R}^{-1}\Gamma\big(\mathbf{Q}^{-1}+\Gamma^{\mathrm{T}}\mathbf{R}^{-1}\Gamma\big)^{-1}\Gamma^{\mathrm{T}}\mathbf{R}^{-1}O\Big) \nonumber \\
  &=H_{11}\big(P^{-1}+O^{\mathrm{T}}\big(\mathbf{R}+\Gamma \mathbf{Q}\Gamma^{\mathrm{T}}\big)^{-1}O\big)\\
   &= H_{11}H_{11}^{-1} \nonumber\\
  &=I\nonumber
\end{align}
By substituting \eqref{matrix_inverse_apply} into \eqref{error1}, the estimation error dynamics are derived as:
\begin{align*}
  e_{k-N} & = H_{11}P^{-1}Ae_{k-N-1} \\
   &= \big(P^{-1}+O^{\mathrm{T}}\big(\mathbf{R}+\Gamma \mathbf{Q}\Gamma^{\mathrm{T}}\big)^{-1}O\big)^{-1}P^{-1}Ae_{k-N-1}
\end{align*}
If $P$, $Q$, and $R$ are tuned such that condition \eqref{stability_no} holds, then the error dynamics are asymptotically stable. $\square$
\end{proof}

\begin{rmk}
Condition \eqref{stability_no} can be satisfied for any matrix $A$. Let us study 2 cases: 1) $\|A\|\leq1$; 2) $\|A\|> 1$.

First, we consider the case when $\|A\|\leq1$. For any positive definite matrix $O^{\mathrm{T}}(\mathbf{R}+\Gamma \mathbf{Q}\Gamma^{\mathrm{T}})^{-1}OP$, it holds that $\|(I+O^{\mathrm{T}}(\mathbf{R}+\Gamma \mathbf{Q}\Gamma^{\mathrm{T}})^{-1}OP)^{-1}\|<1$. For any $\|A\|\leq1$, we have
\begin{align*}
     &\quad\rho\left((P^{-1}+O^{\mathrm{T}}(\mathbf{R}+\Gamma \mathbf{Q}\Gamma^{\mathrm{T}})^{-1}O)^{-1}P^{-1}A\right)  \\
     &\leq\|(P^{-1}+O^{\mathrm{T}}(\mathbf{R}+\Gamma \mathbf{Q}\Gamma^{\mathrm{T}})^{-1}O)^{-1}P^{-1}A\|  \\
     &\leq \|(I+O^{\mathrm{T}}(\mathbf{R}+\Gamma \mathbf{Q}\Gamma^{\mathrm{T}})^{-1}OP)^{-1}\|\|A\|\\
     &<1,
\end{align*}
which ensures that \eqref{stability_no} is satisfied.

Next, we consider the cases when $\|A\|> 1$. Based on spectral decomposition, for real symmetric matrix $O^{\mathrm{T}}(\mathbf{R}+\Gamma \mathbf{Q}\Gamma^{\mathrm{T}})^{-1}O$, there exists an orthogonal matrix $S$ and a diagonal matrix $E=\mathrm{diag}\{\lambda_{1}, \lambda_{2}, \ldots, \lambda_{n_{x}}\}$, such that $O^{\mathrm{T}}(\mathbf{R}+\Gamma \mathbf{Q}\Gamma^{\mathrm{T}})^{-1}O = S^{\mathrm{T}}$ES. $S$ consists of the eigenvectors of $O^{\mathrm{T}}(\mathbf{R}+\Gamma \mathbf{Q}\Gamma^{\mathrm{T}})^{-1}O$ and the diagonal elements $\lambda_{i}$, $i=1,2,\ldots,n_x$, of $E$ are corresponding eigenvalues. Let diagonal matrix $P=\mathrm{diag}\{p_1, p_2, \ldots, p_{n_{x}}\}$. Due to the fact that $S^{\mathrm{T}}=S^{-1}$, one yields
\begin{align*}\label{spectral}
  (P^{-1}+O^{\mathrm{T}}(\mathbf{R}+\Gamma \mathbf{Q}\Gamma^{\mathrm{T}})^{-1}O)^{-1}P^{-1} &=  (I+O^{\mathrm{T}}(\mathbf{R}+\Gamma \mathbf{Q}\Gamma^{\mathrm{T}})^{-1}OP)^{-1} \\
   & = \left(S^{\mathrm{T}}\mathrm{diag}\{1+\lambda_{1}p_{1}, 1+\lambda_{2}p_{2}, \ldots,1+\lambda_{n_{x}}p_{n_{x}}\}S\right)^{-1}\\
   & = S^{\mathrm{T}}\mathrm{diag}\Big\{\frac{1}{1+\lambda_{1}p_{1}}, \frac{1}{1+\lambda_{2}p_{2}}, \ldots,\frac{1}{1+\lambda_{n_{x}}p_{n_{x}}}\Big\}S
\end{align*}
Hence, $\frac{1}{1+\lambda_{i}p_{i}}$, $i=1,2,\ldots, n_{x}$, are the eigenvalues of matrix $(P^{-1}+O^{\mathrm{T}}(\mathbf{R}+\Gamma \mathbf{Q}\Gamma^{\mathrm{T}})^{-1}O)^{-1}P^{-1}$. Since the spectral radius of a square matrix is maximum of the absolute values of its eigenvalues, the increase in the values of the diagonal elements $p_{i}$ will decrease the eigenvalue $\frac{1}{1+\lambda_{i}p_{i}}$ of matrix $(P^{-1}+O^{\mathrm{T}}(\mathbf{R}+\Gamma \mathbf{Q}\Gamma^{\mathrm{T}})^{-1}O)^{-1}P^{-1}$, thus decreasing its spectral radius. Therefore, for the cases when $\|A\|>1$, the diagonal elements of matrix $P$ can be made sufficiently large to satisfy the convergence condition \eqref{stability_no} of DMHE-1. It is worth mentioning that for the cases when $\|A\|>1$, excessively large diagonal elements of matrix $P$ will render the iterative execution of the estimators divergent. Therefore, in such cases when $\|A\|>1$, $P$ should be tuned such that a balance may be achieved between the satisfaction of \eqref{convergence_no} and the satisfaction of \eqref{stability_no}.
\end{rmk}

\subsection{Convergence of the constrained iterative DMHE (DMHE-2)}
In this subsection, the convergence of the proposed DMHE-2 algorithm in \eqref{DMHE2} is analyzed in the presence of constraints on subsystem states and disturbances.


\begin{definition}(Lipschitz Continuity of $\bigtriangledown \Phi$ \cite{bertsekas2015parallel})\label{lipschitz}
For continuous differentiable function $\Phi$, if there exists a Lipschitz constant $K$ such that its gradient $\bigtriangledown\Phi$ satisfies
  \begin{equation}\label{K}
    \|\bigtriangledown \Phi(a)-\bigtriangledown \Phi(b)\|\leq K\|a-b\|, ~\forall a,b\in \mathbb{R},
  \end{equation}
then the gradient $\bigtriangledown \Phi$ is Lipschitz continuous.
\end{definition}

\begin{lemma}(The scaled gradient projection \cite{bertsekas2015parallel})\label{SGP_convergence}
For convex function $\Phi>0$ of which Lipschitz continuous gradient $\bigtriangledown\Phi$ has a Lipschitz constant $K$, given a symmetric and invertible matrix $F$, if there exists $\alpha >0$ such that:
\begin{equation}\label{alpha}
  (a-b)^\mathrm{T}F(a-b)\geq\alpha\|a-b\|^{2}, ~\forall a,b\in\mathbb{Z}
\end{equation}
where $\mathbb{Z}$ is assumed to be a nonempty, closed and convex set, then for any $\gamma$ satisfying $0<\gamma <\frac{2\alpha}{K}$, the sequence of $\hat{z}^{[p]}$ updated by 
\begin{equation}\label{SGP_function}
  \hat{z}^{[p]}=\Big[\hat{z}^{[p-1]}-\gamma F^{-1}\bigtriangledown \Phi(\hat{z}^{[p-1]})\Big]^+,~p=1,2,\ldots
\end{equation}
converges, and ultimately minimizes $\Phi$.

In addition, the solution for the scaled gradient projection scheme in \eqref{SGP_function} is equivalent to the following quadratic optimization problem:
\begin{equation}\label{SGP_centralized}
\hat{z}^{[p]}=\arg\min_{\hat{z}\in\mathbb{Z}}\Big\{\frac{1}{2\gamma}(\hat{z}-\hat{z}^{[p-1]})^{\mathrm{T}}F(\hat{z}-\hat{z}^{[p-1]})+(\hat{z}-\hat{z}^{[p-1]})^{\mathrm{T}}\bigtriangledown\Phi(\hat{z}^{[p-1]})\Big\}
\end{equation}
\end{lemma}



\begin{theorem}\label{thm3}
At each sampling instant $k\geq N$, the estimates of the subsystem state and disturbances (i.e.,  $\hat{x}_{k-N|k}^{i,[p]}$ and $\{\hat{w}^{i,[p]}\}_{k-N}^{k-1}$) provided by DMHE-2 converge to the solution for CMHE \eqref{CMHE_opt} as the estimators are executed for infinite iteration steps, if
\begin{equation}\label{convergence}
    2\lambda_{\min}((\tilde{\mathbf{Q}}^{-1}+\Pi^{\mathrm{T}}{\bf{R}}\Pi)_d)>\lambda_{\max}(\tilde{\mathbf{Q}}^{-1}+\Pi^{\mathrm{T}}{\bf{R}}^{-1}\Pi)
  \end{equation}
  where $\tilde{\mathbf{Q}}=\mathrm{diag}\{\tilde{\mathbf{Q}}_{1}, \tilde{\mathbf{Q}}_{2}, \ldots, \tilde{\mathbf{Q}}_n\}$ with $\tilde{\mathbf{Q}}_i = \mathrm{diag}\{P_i, \mathbf{\mathbf{Q}}_{i}\}$, and $\Pi=[\Pi_{1}, \Pi_{2}, \ldots, \Pi_{n}]$ with $\Pi_i = [O_{[:,i]}, \Gamma_{[:,i]}]$.
\end{theorem}

\begin{proof}
This proof is inspired by the proof of Theorem 3 in \cite{schneider2015convergence}. Define $\hat{z}^{i,[p]}_{k-N} = \left[\begin{array}{c}
                                                                                                       \hat{x}_{k-N|k}^{i,[p]} \\
                                                                                                       \{\hat{w}^{i,[p]}\}_{k-N}^{k-1}
                                                                                                     \end{array}\right]$,
                                                                                                     $\bar{z}_{k-N}^{i}=\left[\begin{array}{c}
                                                                                                                                \bar{x}_{k-N}^{i} \\
                                                                                                                                \mathbf{0}
                                                                                                                              \end{array}\right]$, where $\mathbf{0}$ is the zero vector of which the dimension is the same as $\{\hat{w}^{i,[p]}\}_{k-N}^{k-1}$.
The objective function $\bar{\Phi}_{k}^{i,[p]}$ can be rewritten as
\begin{align}\label{PMHE_obj_extended}
   \bar{\Phi}_{k}^{i,[p]} &=   \frac{1}{2}\Big(\|\hat{z}^{i,[p]}_{k-N}\|^2_{\tilde{\mathbf{Q}}_i^{-1}}+\big\|\{y\}_{k-N}^{k}-\Pi_i \hat{z}^{i,[p]}_{k-N}-\sum_{l\in\mathbb{N}\setminus\{i\}}\Pi_l \hat{z}^{l,[p-1]}_{k-N}-\Lambda\{u\}_{k-N}^{k-1}\big\|^{2}_{{\bf{R}}^{-1}}\Big)\nonumber\\
   &\quad+\frac{1}{2}\bar{z}^{i,\mathrm{T}}_{k-N}\tilde{\mathbf{Q}}^{-1}_{i}\bar{z}^{i}_{k-N}-\hat{z}^{i,[p],\mathrm{T}}_{k-N}\tilde{\mathbf{Q}}^{-1}_{i}\bar{z}^{i}_{k-N}\\
   &=  \frac{1}{2}\hat{z}_{k-N}^{i,[p],\mathrm{T}}(\tilde{\mathbf{Q}}^{-1}_{i}+\Pi_{i}^{\mathrm{T}}{\bf{R}}^{-1}\Pi_{i})\hat{z}_{k-N}^{i,[p]}-\hat{z}_{k-N}^{i,[p],\mathrm{T}}\Pi_{i}^{\mathrm{T}}{\bf{R}}^{-1}\Big(\{y\}_{k-N}^{k}-\sum_{l\in\mathbb{N}\setminus\{i\}}\Pi_l \hat{z}_{k-N}^{l,[p-1]}-\Lambda\{u\}_{k-N}^{k-1}\Big)\nonumber\\
   &\quad-\hat{z}_{k-N}^{i,[p],\mathrm{T}}\tilde{\mathbf{Q}}^{-1}_{i}\bar{z}_{k-N}^{i}+\frac{1}{2}\bar{z}^{i,\mathrm{T}}_{k-N}\tilde{\mathbf{Q}}^{-1}_{i}\bar{z}^{i}_{k-N}+\frac{1}{2}\{y\}_{k-N}^{k,\mathrm{T}}{\bf{R}}^{-1}\{y\}_{k-N}^{k}\nonumber\\
   &\quad+\frac{1}{2}\Big(\sum_{l\in\mathbb{N}\setminus\{i\}}\hat{z}_{k-N}^{l,[p-1],\mathrm{T}}\Pi_{l}^{\mathrm{T}}\Big){\bf{R}}^{-1}\Big(\sum_{l\in\mathbb{N}\setminus\{i\}}\Pi_{l}\hat{z}_{k-N}^{l,[p-1]}\Big)+\frac{1}{2}\{u\}_{k-N}^{k-1,\mathrm{T}}\Lambda^{\mathrm{T}}{\bf{R}}^{-1}\Lambda\{u\}_{k-N}^{k-1}\nonumber\\
   &\quad-\{y\}_{k-N}^{k,\mathrm{T}}{\bf{R}}^{-1}\Big(\sum_{l\in\mathbb{N}\setminus\{i\}}\Pi_{l}\hat{z}_{k-N}^{l,[p-1]}+\Lambda\{u\}_{k-N}^{k-1}\Big)+\Big(\sum_{l\in\mathbb{N}\setminus\{i\}}\hat{z}_{k-N}^{l,[p-1],\mathrm{T}}\Pi_l^{\mathrm{T}}\Big){\bf{R}}^{-1}\Lambda\{u\}_{k-N}^{k-1}\nonumber
\end{align}
At every time instant $k$, for \eqref{SGP_centralized} in Lemma \ref{SGP_convergence}, if we replace $\hat{z}^{[p-1]}$ with $\hat{z}^{[p-1]}_{k-N}$, $\hat{z}$ with $\hat{z}_{k-N}^{[p]}$, and $F$ with a block diagonal matrix with main diagonals being $F_i$, then \eqref{SGP_centralized} can be rewritten as
\begin{equation}\label{SGP_medium} \hat{z}^{[p]}_{k-N}=\arg\min_{\hat{z}_{k-N}^{[p]}\in\mathbb{Z}}\sum_{i=1}^{n}\Big\{\frac{1}{2\gamma}(\hat{z}^{i,[p]}_{k-N}-\hat{z}^{i,[p-1]}_{k-N})^{\mathrm{T}}F_i(\hat{z}^{i,[p]}_{k-N}-\hat{z}_{k-N}^{i,[p-1]})+(\hat{z}_{k-N}^{i,[p]}-\hat{z}_{k-N}^{i,[p-1]})^{\mathrm{T}}\bigtriangledown_{i} \Phi(\hat{z}_{k-N}^{[p-1]})\Big\}
\end{equation}
where $\hat{z}_{k-N}^{[p]}=\mathrm{col}(\hat{z}_{k-N}^{1,[p]},\hat{z}_{k-N}^{2,[p]},\ldots,\hat{z}_{k-N}^{n,[p]})$; $\bigtriangledown_{i}\Phi(z)$ contains the elements of the partial derivative vector $\bigtriangledown\Phi$ that correspond to $\hat{z}^{i}$. For non-empty, closed, and convex $\mathbb{Z}=\mathbb{Z}^{1}\times\mathbb{Z}^{2}\times\ldots\times\mathbb{Z}^{n}$, parallelization of \eqref{SGP_medium} is feasible, and we can obtain
\begin{equation}\label{SGP_distributed} \hat{z}_{k-N}^{{i},[p]}=\arg\min_{\hat{z}_{k-N}^{i,[p]}\in\mathbb{Z}^{i}}\Big\{\frac{1}{2\gamma}(\hat{z}_{k-N}^{i,[p]}-\hat{z}_{k-N}^{i,[p-1]})^{\mathrm{T}}F_i(\hat{z}_{k-N}^{i,[p]}-\hat{z}_{k-N}^{i,[p-1]})+(\hat{z}_{k-N}^{i,[p]}-\hat{z}_{k-N}^{i,[p-1]})^{\mathrm{T}}\bigtriangledown_{i} \Phi(\hat{z}^{[p-1]}_{k-N})\Big\}
\end{equation}
Choose $\Phi = \Phi_{k}$ in \eqref{CMHE_obj}, one yields
\begin{align}\label{nabla}
  \bigtriangledown_{i}\Phi_{k}(\hat{z}_{k-N}^{[p-1]})&=(\tilde{\mathbf{Q}}^{-1}_{i}+\Pi_{i}^{\mathrm{T}}{\bf{R}}^{-1}\Pi_{i})\hat{z}_{k-N}^{i,[p-1]}-\tilde{\mathbf{Q}}^{-1}_{i}\bar{z}_{k-N}^{i}\nonumber\\
  &\quad-\Pi_{i}^{\mathrm{T}}{\bf{R}}^{-1}\Big(\{y\}_{k-N}^{k}-\sum_{l\in\mathbb{N}\setminus\{i\}}\Pi_l \hat{z}_{k-N}^{l,[p-1]}-\Lambda\{u\}_{k-N}^{k-1}\Big)
\end{align}
For \eqref{SGP_distributed}, let $\gamma=1$ and $F=\tilde{\mathbf{Q}}^{-1}+(\Pi^{\mathrm{T}}{\bf{R}}^{-1}\Pi)_{d}$ with $F_i = \tilde{\mathbf{Q}}_i^{-1}+\Pi_{i}^{\mathrm{T}}{\bf{R}}^{-1}\Pi_{i}$.
Substituting \eqref{nabla} into \eqref{SGP_distributed} yields
\begin{align}\label{sgd}
  \hat{z}_{k-N}^{i,[p]} &= \mathrm{arg}\min_{\hat{z}_{k-N}^{i,[p]}\in\mathbb{Z}^{i}}\Big\{\frac{1}{2}\hat{z}_{k-N}^{i,[p],\mathrm{T}}(\tilde{\mathbf{Q}}^{-1}_i+\Pi_i^{\mathrm{T}}{\bf{R}}^{-1}\Pi_i)\hat{z}_{k-N}^{i,[p]}-\hat{z}_{k-N}^{i,[p],\mathrm{T}}\tilde{\mathbf{Q}}^{-1}_{i} \bar{z}_{k-N}^{i}  \nonumber\\
   &\quad  -\hat{z}_{k-N}^{i,[p],\mathrm{T}}\Pi_i^{\mathrm{T}}{\bf{R}}^{-1}\Big(\{y\}_{k-N}^{k}-\sum_{l\in\mathbb{N}\setminus\{i\}}\Pi_l \hat{z}_{k-N}^{l,[p-1]}-\Lambda\{u\}_{k-N}^{k-1}\Big) \nonumber \\
   &\quad+\hat{z}_{k-N}^{i,[p-1],\mathrm{T}}\Pi_{i}^{\mathrm{T}}{\bf{\mathbf{Q}}}^{-1}\Big(\{y\}_{k-N}^{k}-\sum_{l\in\mathbb{N}\setminus\{i\}}\Pi_{l}\hat{z}_{k-N}^{l,[p-1]}-\Lambda\{u\}_{k-N}^{k-1}\Big) \\
   &\quad-\frac{1}{2}\hat{z}_{k-N}^{i,[p-1],\mathrm{T}}(\tilde{\mathbf{Q}}^{-1}_i+\Pi_{i}^{\mathrm{T}}{\bf{R}}^{-1}\Pi_{i})\hat{z}_{k-N}^{i,[p-1]}+\hat{z}_{k-N}^{i,[p-1],\mathrm{T}}\tilde{\mathbf{Q}}^{-1}_{i}\bar{z}_{k-N}^{i}\Big\}\nonumber
\end{align}
The objective function of \eqref{sgd} is equivalent to $\bar{\Phi}_{k}^{i,[p]}$ of DMHE-2 algorithm defined in \eqref{PMHE_obj_extended} up to an additional constant.
Based on the relation between DMHE-2 algorithm and the scaled gradient projection algorithm, we can establish the convergence condition for DMHE-2. According to Lemma \ref{SGP_convergence}, the scaled gradient projection algorithm converges if $0<\gamma <\frac{2\alpha}{K}$. By choosing $\alpha=\lambda_{\min}(\tilde{R}^{-1}+\Pi^{\mathrm{T}}{\bf{Q}}^{-1}\Pi)_d$, \eqref{alpha} in Lemma \ref{SGP_convergence} can be satisfied. Furthermore, the objective function $\Phi_{k}$ in \eqref{DMHE_obj1} is convex and the Lipschitz condition \eqref{K} in Lemma \ref{SGP_convergence} can be satisfied by selecting $K=\lambda_{\max}(\tilde{R}^{-1}+\Pi^{\mathrm{T}}{\bf{Q}}^{-1}\Pi)$.  Consequently, the convergence condition \eqref{convergence} derived from the scaled gradient projection serves as the convergence condition of DMHE-2. $\square$

\end{proof}

\subsection{Stability of the constrained DMHE (DMHE-2)}
In this subsection, we prove the stability of the estimation error dynamics for the entire DMHE scheme.
\begin{theorem}
If Assumptions \ref{assume1} and \ref{assume2} hold, if the MHE-based estimators of DMHE-2 are executed for infinite iteration steps at each sampling instant $k$, and if tuning parameters $P$, $Q$, and $R$ are selected such that
\begin{equation}\label{DMHE_con_stability}
  \frac{8\lambda_{\min}(A^{\mathrm{T}}P^{-1}A)}{\lambda_{\max}(P^{-1}+O^{\mathrm{T}}(\mathbf{R}+\frac{1}{2}\Gamma \mathbf{Q}\Gamma^{\mathrm{T}})^{-1}O)}<1,
\end{equation}
then DMHE-2 is asymptotically stable.
\end{theorem}
\begin{proof}
First, we establish an upper bound on optimal objective function $\Phi_{k}^{o}$ which has the following form:
\begin{equation}\label{optimal_objective}
  \Phi_{k}^{o}=\frac{1}{2}\|\hat{x}_{k-N|k}-\bar{x}_{k-N}\|^{2}_{P^{-1}}+\frac{1}{2}\sum_{j=k-N}^{k-1}\|\hat{w}_{j}\|_{Q^{-1}}^{2}+\frac{1}{2}\sum_{j=k-N}^{k}\|\hat{v}_{j}\|_{R^{-1}}^{2}
\end{equation}
According to the optimality of $\hat{x}_{k-N|k}$, it follows that $\Phi_{k}^{o}\leq\Phi_{k}(x_{k-N})$, and we have
\begin{align}\label{upper_bound}
  \Phi_{k}^{o} & \leq\frac{1}{2}\|x_{k-N}-\bar{x}_{k-N}\|_{P^{-1}}^{2}+\frac{1}{2}\sum_{j=k-N}^{k-1}\|w_{j}\|^2_{Q^{-1}}+ \frac{1}{2}\sum_{j=k-N}^{k}\|v_{j}\|^2_{R^{-1}}  \\
   & \leq\frac{1}{2}\left(\|x_{k-N}-\bar{x}_{k-N}\|_{P^{-1}}^{2}+\rho_1\right)\nonumber
\end{align}
where $\rho_1=Nr_{w}^{2}\lambda_{\max}(Q^{-1})+(N+1)r_{v}^{2}\lambda_{\max}(R^{-1})$ with $r_{w}$ and $r_{v}$ being the upper bounds on the disturbances and noise as defined in Assumption~\ref{assume1}. In addition, it holds that
\begin{equation}\label{triangle}
  \|a_1\|^{2}\geq\frac{1}{2}\|a_1+a_2\|^{2}-\|a_{2}\|^{2}
\end{equation}
for any vectors $a_{1}$ and $a_{2}$.

Next, we establish a lower bound for the optimal cost $\Phi_{k}^{o}$ defined in \eqref{optimal_objective}.
By exploiting inequality \eqref{triangle}, it yields that
\begin{equation}\label{x}
  \|\hat{x}_{k-N|k}-\bar{x}_{k-N}\|^{2}_{P^{-1}}\geq\frac{1}{2}\|x_{k-N}-\hat{x}_{k-N|k}\|^{2}_{P^{-1}}-\|x_{k-N}-\bar{x}_{k-N}\|^{2}_{P^{-1}}
\end{equation}
By considering \eqref{model} and \eqref{triangle}, one can have
\begin{align}\label{y}
 \sum_{j=k-N}^{k}\|\hat{v}_{k}\|^{2}_{R^{-1}}  &=\left\|\{y\}_{k-N}^{k}-O\hat{x}_{k-N|k}-\Gamma\{\hat{w}\}_{k-N}^{k-1}-\Lambda\{u\}_{k-N}^{k-1}\right\|^2_{{\bf{R}}^{-1}}  \nonumber\\
   &= \|O(x_{k-N}-\hat{x}_{k-N|k})+\Gamma(\{w\}^{k-1}_{k-N}-\{\hat{w}\}_{k-N}^{k-1})+\{v\}_{k-N}^{k}\|^2_{{\bf{R}}^{-1}}  \\
   &\geq \frac{1}{2}\|O(x_{k-N}-\hat{x}_{k-N|k})-\Gamma\{\hat{w}\}_{k-N}^{k-1}\|^2_{{\bf{R}}^{-1}}-\|\Gamma\{w\}^{k-1}_{k-N}+\{v\}_{k-N}^{k}\|^2_{\bf{R}^{-1}}\nonumber\\
   & \geq\frac{1}{2}\|O(x_{k-N}-\hat{x}_{k-N|k})-\Gamma\{\hat{w}\}_{k-N}^{k-1}\|^2_{{\bf{R}}^{-1}}-\rho_{2}\nonumber
\end{align}
where $\rho_2=\sup_{k}\|\Gamma\{w\}^{k-1}_{k-N}+\{v\}_{k-N}^{k}\|^2_{{\bf{R}}^{-1}}$.

Based on inequality (30) in \cite{battistelli2018distributed}, we have
\begin{equation}\label{y_and_w}
   \frac{1}{2}\|O(x_{k-N}-\hat{x}_{k-N|k})-\Gamma\{\hat{w}\}_{k-N}^{k-1}\|^2_{\bf{R}^{-1}}+\|{\{\hat{w}\}}_{k-N}^{k-1}\|_{{\bf {Q}}^{-1}}^{2}\geq\frac{1}{2}\|x_{k-N}-\hat{x}_{k-N|k}\|^{2}_{G}
\end{equation}
where $ G  = O^{\mathrm{T}}({\bf{R}}^{-1}-{\bf{R}}^{-1}\Gamma(\Gamma^{\mathrm{T}}{\bf{R}}^{-1}\Gamma+2{\bf { Q}}^{-1})^{-1}\Gamma^{\mathrm{T}}{\bf{R}}^{-1})O  = O^{\mathrm{T}}({\bf{R}}+\frac{1}{2}\Gamma {\bf { Q}}\Gamma^{\mathrm{T}})^{-1}O$.
Substituting \eqref{x}, \eqref{y}, and \eqref{y_and_w} in the optimal objective function \eqref{optimal_objective}, one can derive the following inequality
\begin{align}\label{lower_bound}
 \Phi_{k}^{o}  & \geq \frac{1}{4}\|x_{k-N}-\hat{x}_{k-N|k}\|^{2}_{P^{-1}}-\frac{1}{2}\|x_{k-N}-\bar{x}_{k-N}\|^{2}_{P^{-1}}+\frac{1}{4}\|x_{k-N}-\hat{x}_{k-N|k}\|^{2}_{G}-\frac{1}{2}\rho_{2}\nonumber \\
   & =\frac{1}{4}\|x_{k-N}-\hat{x}_{k-N|k}\|^{2}_{\Xi}-\frac{1}{2}\|x_{k-N}-\bar{x}_{k-N}\|^{2}_{P^{-1}}-\frac{1}{2}\rho_{2}
\end{align}
where $\Xi =P^{-1}+G $.

We leverage the upper bound \eqref{upper_bound} and lower bound \eqref{lower_bound} on $\Phi_{k}^{o}$ to obtain an upper bound on the estimation error for the entire system as follows:
\begin{equation}\label{bound of e}
 \|x_{k-N}-\hat{x}_{k-N|k}\|^{2}_{\Xi}\leq4\|x_{k-N}-\bar{x}_{k-N}\|_{P^{-1}}^{2}+2(\rho_1+\rho_2)
\end{equation}
Base on the fact that $\bar{x}_{k-N}=A\hat{x}_{k-N-1|k-1}+Bu_{k-N-1}$ and \eqref{model1}, one can derive
\begin{align}\label{error}
  \|e_{k-N}\|^{2}_{\Xi} & \leq4\|Ax_{k-N-1}+w_{k-N-1}-A\hat{x}_{k-N-1|k-1}\|_{P^{-1}}^{2}+2(\rho_1+\rho_2) \nonumber\\
  &\leq 8\|e_{k-N-1}\|^{2}_{\Omega}+8\|w_{k-N-1}\|_{P^{-1}}^{2}+2(\rho_1+\rho_2)\nonumber\\
   & \leq 8\|e_{k-N-1}\|^{2}_{\Omega}+\rho\nonumber
\end{align}
with $\Omega=A^{\mathrm{T}}P^{-1}A$ and $\rho = 8\lambda_{\max}(P^{-1})r_{w}+ 2(\rho_1+\rho_2)$.
Then if $w_{k}=0$, $v_{k}=0$ for $k=N, N+1, \ldots$, the norm of the estimation error is bounded as
\begin{equation*}
  \|e_{k-N}\|^{2}_{\Xi}  \leq8\|e_{k-N-1}\|^{2}_{\Omega}, ~k=N, N+1, \ldots
\end{equation*}
Moreover, if $P$, $Q$, and $R$ are tuned such that \eqref{DMHE_con_stability} is satisfied, then $\lim_{k\rightarrow\infty}\|e_{k-N}\|^{2}=0$. $\square$

\end{proof}
\begin{rmk}
  If the number of iteration steps $p\rightarrow\infty$, the solutions for DMHE-2 and CMHE in \eqref{CMHE_opt} are equivalent according to Theorem \ref{thm3}. Therefore, the stability criteria of the CMHE can apply to DMHE-2 for asymptotically stable estimation error dynamics.
\end{rmk}

\section{Case study on a simulated chemical process}\label{sec.5}
\subsection{Process description}\label{sim:description}
We consider a reactor-separator chemical process that consists of two continuous stirred tank reactors (CSTRs) and one flash tank separator.
The three vessels as three subsystems are interconnected via mass and energy flows. A schematic of this process is shown in Figure \ref{CSTR}.
\begin{figure}
  \centering
  \includegraphics[width=0.75\textwidth]{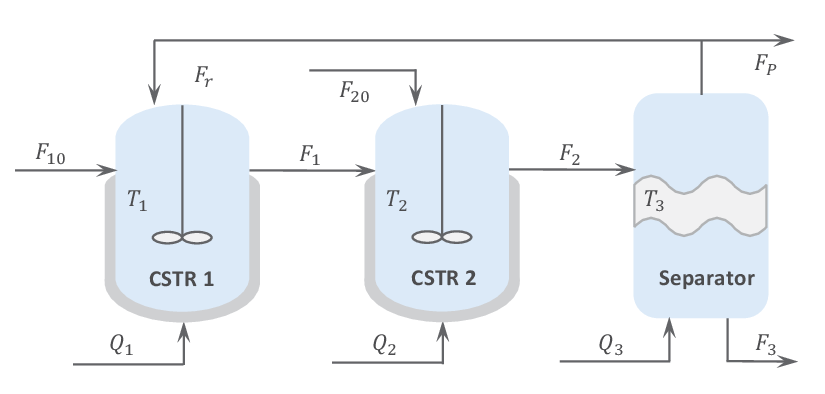}
  \caption{A schematic of the reactor-separator process.}\label{CSTR}
\end{figure}

A fresh feed flow containing pure reactant A is fed into the first vessel at volumetric flow rate $F_{10}$ and temperature $T_{10}$. In this vessel, the primary irreversible reaction converts the reactant A to the desired product B, a portion of which is further converted to undesired side product C. The outlet stream of the first vessel enters the second vessel at flow rate $F_{1}$ and temperature $T_{1}$. Another feed stream carrying pure A enters the second vessel at flow rate $F_{20}$ and temperature $T_{20}$, where the same reactions take place. The effluent of the second vessel is fed to the separator at flow rate $F_{2}$ and temperature $T_{2}$. A recycle stream leaves the separator and enters the first vessel at flow rate $F_{r}$ and temperature $T_{3}$. The three vessels are equipped with heating jackets, each of which can add/remove heat to/from the corresponding vessel at heating input rate $Q_i$, $i=1,2,3$. Based on the mass and energy balances, a first-principle dynamic model has been established \cite{zhang2013distributed}:
\begin{subequations}\label{eqn:exa}
\begin{align}
& \displaystyle{\frac{dx_{A1}}{dt}} = \frac{F_{10}}{V_1}(x_{A10} - x_{A1}) + \frac{F_r}{V_1}(x_{Ar} - x_{A1}) - k_1 e^{\frac{-E_1}{RT_1 }}x_{A1}  \\[0.3em]
& \displaystyle{\frac{dx_{B1}}{dt}} = \frac{F_{10}}{V_1}(x_{B10} - x_{B1}) +  \frac{F_r}{V_1}(x_{Br} - x_{B1})+  k_1 e^{\frac{-E_1}{RT_1 }}x_{A1} - k_2 e^{\frac{-E_2}{RT_1 }}x_{B1}  \\[0.3em]
& \displaystyle{~~\frac{dT_{1}}{dt}} = \frac{F_{10}}{V_1}(T_{10} - T_{1}) + \frac{F_r}{V_1}(T_{3} - T_{1}) - \frac{\Delta H_1}{ c_p}k_{1}e^{\frac{-E_1}{RT_1}}x_{A1} - \frac{\Delta H_2}{ c_p}k_{2}e^{\frac{-E_2}{RT_1}}x_{B1} +  \frac{Q_1}{\rho c_pV_1} \\[0.3em]
& \displaystyle{\frac{dx_{A2}}{dt}} = \frac{F_{1}}{V_2}(x_{A1} - x_{A2}) + \frac{F_{20}}{V_2}(x_{A20} - x_{A2}) - k_1 e^{\frac{-E_1}{RT_2 }}x_{A2}  \\[0.3em]
& \displaystyle{\frac{dx_{B2}}{dt}} = \frac{F_{1}}{V_2}(x_{B1} - x_{B2}) + \frac{F_{20}}{V_2}(x_{B20} - x_{B2})+  k_1 e^{\frac{-E_1}{RT_2 }}x_{A2} - k_2 e^{\frac{-E_2}{RT_2}}x_{B2}  \\[0.3em]
& \displaystyle{~~\frac{dT_{2}}{dt}} = \frac{F_{1}}{V_2}(T_{1} - T_{2}) + \frac{F_{20}}{V_2}(T_{20} - T_{2}) - \frac{\Delta H_1}{ c_p}k_{1}e^{\frac{-E_1}{RT_2}}x_{A2} - \frac{\Delta H_2}{ c_p}k_{2}e^{\frac{-E_2}{RT_2}}x_{B2} +  \frac{Q_2}{\rho c_pV_2} \\[0.3em]
& \displaystyle{\frac{dx_{A3}}{dt}} = \frac{F_{2}}{V_3}(x_{A2} - x_{A3}) - \frac{(F_r+F_p)}{V_3}(x_{Ar} - x_{A3})  \\[0.3em]
& \displaystyle{\frac{dx_{B3}}{dt}} = \frac{F_{2}}{V_3}(x_{B2} - x_{B3}) - \frac{(F_r+F_p)}{V_3}(x_{Br} - x_{B3}) \\[0.3em]
& \displaystyle{~~\frac{dT_{3}}{dt}} = \frac{F_{2}}{V_3}(T_{2} - T_{3}) + \frac{Q_3}{\rho c_pV_3}+\frac{(F_{r}+F_{p})}{\rho c_{p}V_{3}}(x_{Ar}\Delta H_{\text{vap1}}+x_{Br}\Delta H_{\text{vap2}}+x_{Cr}\Delta H_{\text{vap3}} )
\end{align}
\end{subequations}
The algebraic equations that characterize the compositions of the overhead stream are given as follows:
\begin{equation}\label{eqn:x_r}
\begin{aligned}
x_{Ar} &= \frac{\alpha _A x_{A3}}{\alpha _A x_{A3} + \alpha _B x_{B3} + \alpha _C x_{C3} } \\[0.3em]
x_{Br} &= \frac{\alpha _B x_{B3}}{\alpha _A x_{A3} + \alpha _B x_{B3} + \alpha _C x_{C3}} \\[0.3em]
x_{Cr} &= \frac{\alpha _C x_{C3}}{\alpha _A x_{A3} + \alpha _B x_{B3} + \alpha _C x_{C3}}\\[0.3em]
x_{C3} &= 1-x_{A3}-x_{B3}
\end{aligned}
\end{equation}
%
The definitions of the process states are presented in Table~\ref{tbl:variables}. 
The temperatures in the three vessels are measured online using hardware sensors. The objective is to apply the proposed methods to develop DMHE schemes that consist of three local estimators for the three subsystems, which will be used to estimate the nine system states based on the sensor measurements of the temperatures $T_1$, $T_2$, and $T_3$.
More detailed definitions of the parameters involved in the model in \eqref{eqn:exa} and \eqref{eqn:x_r} can be found in \cite{zhang2013distributed}, while the values of the parameters are shown in Table~\ref{tbl:parameter}.
\begin{table}
  \centering
    \caption{Definitions of process states}\label{tbl:variables}
  \begin{tabular}{cccccc}
  \toprule
  State & Physical meaning \\
  \midrule
  $x_{A1}$, $x_{A2}$, $x_{A3}$  &  Mass fractions of $A$ in vessels 1,2,3  \\
  \hline
  $x_{B1}$, $x_{B2}$, $x_{B3}$ &   Mass fractions of $B$ in vessels 1,2,3\\
  \hline
  $T_{1}$, $T_{2}$, $T_{3}$ &    Temperatures of feed stream in vessels 1,2,3\\
  \bottomrule
\end{tabular}
\end{table}
We consider a steady-state operating point $x_s = [0.2055~~0.6751~~474.0056~\mathrm{K}~~0.2243~~0.6564~~467.2124~\mathrm{K}$ $0.0781~~0.7032~~468.9572~\mathrm{K}]^{\mathrm{T}}$, which is corresponding to constant inputs $Q_1 = 2.8\times 10^6~\mathrm{kJ/h}$, $Q_2=1.1\times 10^6~\mathrm{kJ/h}$, $Q_3=2.8\times 10^6~\mathrm{kJ/h}$, $F_{10}=5.04~\mathrm{m^3/h}$, $F_{20}=5.04~\mathrm{m^3/h}$, $F_{r}=50.4~\mathrm{m^3/h}$, and $F_{p}=0.504~\mathrm{m^3/h}$. At $x_s$, a linearized model is obtained through Taylor-series expansion. The linearized model is further discretized with discretization step size $h=0.02~\text{h}$, and a discrete-time model in the form of \eqref{model} can be obtained for the process. The system matrices of the discrete-time linearized model are presented in the Appendix.
In the implementation, considering the significant difference in the magnitudes of the states, we scale each of the states through dividing it by the maximum value that the corresponding state can take during the entire process operation. Since all the states of the considered chemical process are non-negative, the scaled states stay within $[0,1]$. Accordingly, state estimation is conducted in the scaled coordinate. This way, potential numerical issues can be avoided. Addictive unknown system disturbances and measurement noise that follow the Gaussian distribution with zero mean and standard deviation of 0.01 are added to the scaled states and measured outputs, respectively. The estimates of the actual process states can be obtained based on the estimates in the scaled coordinate.

\begin{table}
\renewcommand\arraystretch{1.25}
  \centering
    \caption{Process parameters}\label{tbl:parameter}
  \begin{tabular}{llllll}
   \toprule[1pt]
   & $V_{1} = 1.0~ \mathrm{m^3/h}$  \quad\quad\quad\quad\quad& $\Delta H_{1} = -6.0\times 10^{4} ~\mathrm{kJ/kmol}$\\
&  $V_{2} = 0.5 ~\mathrm{m^3/h}$\quad\quad\quad\quad\quad& $\Delta H_{2} = -7.0\times 10^{4} ~\mathrm{kJ/kmol}$\\
& $V_{3} = 1.0 ~\mathrm{m^3/h}$\quad\quad\quad\quad\quad& $\Delta H_{\text{vap1}} = -3.53\times 10^{4} ~\mathrm{kJ/kmol}$\\
& $\alpha_{A} = 3.5$ \quad\quad\quad\quad\quad& $\Delta H_{\text{vap2}} = -1.57\times 10^{4} ~\mathrm{kJ/kmol}$\\
& $\alpha_{B} = 1.0$\quad\quad\quad\quad\quad& $\Delta H_{\text{vap3}} = -4.068\times 10^{4} ~\mathrm{kJ/kmol}$\\
&$\alpha_{C} = 0.5$  \quad\quad\quad\quad\quad& $k_1=2.77\times 10^3~ \text{s}^{-1}$\\
& $T_{10} = 300~\mathrm{K}$ \quad\quad\quad\quad\quad& $k_2=2.6\times 10^3 ~\text{s}^{-1}$\\
&$T_{20} = 300~\mathrm{K}$ \quad\quad\quad\quad\quad&  $c_{p}=4.2 ~\mathrm{kJ/kg.K}$\\
& $E_{1} = 5.0\times 10^4~\mathrm{kJ/kmol}$\quad\quad\quad\quad\quad& $x_{A10}=1$\\
& $E_{2} = 6.0\times 10^4~\mathrm{kJ/kmol}$\quad\quad\quad\quad\quad& $x_{B10}=0$ \\
&$R=8.314~ \mathrm{kJ/kmol.K}$\quad\quad\quad\quad\quad&$x_{A20}=1$   \\
&$\rho=1000.0~\mathrm{kg/m^3}$\quad\quad\quad\quad\quad& $x_{B20}=0$\\
   \bottomrule[1pt]
\end{tabular}
\end{table}

\begin{rmk}
A linearized model can approximate the dynamical behavior of the process satisfactorily within a certain region around the operating point where linearization is conducted.
Many chemical processes, including the reactor-separator process considered in this paper, tend to be operated at a steady-state operating point for safe operation and consistent production. For a process of which the operation does not deviate much from the linearization point, the mismatch between the nonlinear model and the linearized model may be acceptable.
In such cases (when the mismatch is not significant), the proposed linear method may be applied. Meanwhile, when the fluctuations in the states become significant, then nonlinear distributed state estimation can be necessary. We note that the proposed method has the potential to be extended to address nonlinear dynamics explicitly; this is left to our future work.
\end{rmk}

\begin{rmk}
It is worth mentioning that state estimation is conducted in the scaled coordinate. In addition, the states of the linearized model represent the deviation of the actual states from the steady-state level. Therefore, the state-steady values of the states need to be added to the estimates generated by the developed distributed estimation scheme to obtain the estimates of the actual states.
\end{rmk}
\subsection{Results for DMHE-1}\label{sim:DMHE-1}
In this section, we present the estimation results given by the proposed DMHE-1, and compare the performance of two DMHE algorithms: (1) the proposed DMHE algorithm, and (2) the DMHE algorithm proposed in \cite{schneider2015convergence}. Both algorithms are partition-based DMHE with an iterative evaluation of local estimators within every sampling interval. Meanwhile, the algorithm in \cite{schneider2015convergence} does not include penalty system disturbances in the objective functions of the local estimators.

In this set of simulations, 
the estimators for the subsystems are required to be executed for 5 iteration steps within every sampling interval. For the $i$th estimator of DMHE-1, the weighting matrices $Q_{i}$ in each estimator and $R$ are chosen as the covariance matrices of the subsystem disturbances and measurement noise added in the simulations; $P_{i}$ in each estimator is a diagonal matrix with main diagonal elements being 0.1 (considering the fact that all the states are scaled in state estimation). The initial state and the initial guess used by the estimators are shown in Table~\ref{tbl:initial_states}. We generate system disturbances and measurement noise that follow Gaussian distribution of zero mean and standard deviation of 0.01, which are added to the scaled system states and measured outputs, respectively. 

\begin{table}
  \centering
    \caption{The actual values $x_{0}$ and the initial guess $\hat{x}_{0}$ of the initial states of the process}\label{tbl:initial_states}
  \begin{tabular}{cccccccccc}
   \toprule[1pt]
   & $x_{A1}$ & $x_{B1}$ & $T_{1}$ $(K)$ & $x_{A2}$ & $x_{B2}$ & $T_{2}$ $(K)$ & $x_{A3}$ & $x_{B3}$ & $T_{3}$ $(K)$\\
   \midrule
   $x_{0}$ & 0.1939 & 0.7404 & 528.3482 & 0.2162 & 0.7190 & 520.0649 & 0.0716 & 0.7373 & 522.3765\\
    $\hat{x}_{0}$ & 0.2080 & 0.7943 & 566.7735 & 0.2319 & 0.7712  & 557.8878 & 0.0768 & 0.7910 & 560.3675\\
   \bottomrule[1pt]
\end{tabular}
\end{table}

\begin{table}[tt]
  \centering
  \setlength{\abovecaptionskip}{0cm}
  \setlength{\belowcaptionskip}{0.2cm}
    \caption{Average RMSEs of the proposed DMHE-1 for different lengths of the estimation horizon.}\label{tbl:horizon}
  \begin{tabular}{c|c|c|c|c|ccccc}
   \hline
   $N$ & 2 & 5 & 10 & 15 & 20\\
   \hline
   RMSE & 0.3123 & 0.2962 & 0.2655 & 0.2770 & 0.2954 \\
   \hline
\end{tabular}
\end{table}

\begin{table}[tt]
  \centering
  \setlength{\abovecaptionskip}{0cm}
  \setlength{\belowcaptionskip}{0.2cm}
 \caption{Average computation time for one-time execution for the two DMHE designs with different lengths of the estimation horizon.}\label{tbl:time}
  \begin{tabular}{c|c|c|c|c|ccccc}
   \hline
   $N$ & 2 & 5 & 10 & 15 & 20\\
   \hline
  DMHE-1 (sec) &  $7.83\times10^{-5}$&  $1.291\times10^{-4}$& $2.217\times10^{-4}$ &  $5.006\times10^{-4}$ & $7.966\times 10^{-4}$  \\
   \hline
   DMHE-2 (sec) &  0.0569& 0.1674 & 0.5358& 1.0787 & 2.3835 \\
   \hline
\end{tabular}
\end{table}

Next, we evaluate the estimation performance of the proposed DMHE-1 with different lengths of estimation horizon $N$. In particular, we consider the cases when $N=2,5,10,15,20$. The estimation performance is evaluated by calculating the root mean square errors (RMSE) over 500 Monte Carlo tests. For each simulation run, the RMSE is computed following:
\begin{equation*}
    \mathrm{RMSE}(k) = \sqrt{\frac{\sum_{j=1}^{n}  \big\|\hat{x}_{k}^{j}-x_{k}^{j}\big\|  ^2}{n}}
\end{equation*}
where $n$ is the number of subsystems and $k$ is the sampling instant. 100 simulation runs are conducted for each considered length of the estimation horizon. Each simulation run covers 2-hour process operation. The average RMSE is computed for each $N$ over the extra 500 simulation runs. The average RMSE values for the five cases (i.e., $N=2,5,10,15,20)$ are presented in Table~\ref{tbl:horizon}. Among the five cases, when $N=10$, DMHE-1 provides the smallest estimation error.
The actual states and the estimates of the states of the process given by provided by DMHE-1 with $N=10$ are presented in Figure \ref{1}. The estimates are able to track the actual system states accurately.
The average computation time for one-time execution for each estimator of DMHE-1 with different estimation horizon lengths is computed based on 100 simulation runs. The largest average computation time among the three estimators of DMHE-1 is presented in Table \ref{tbl:time}. As the length of the estimation horizon increases, the computation time increases accordingly.

\begin{figure}[tttt]
  \centering
  \includegraphics[width=0.95\textwidth]{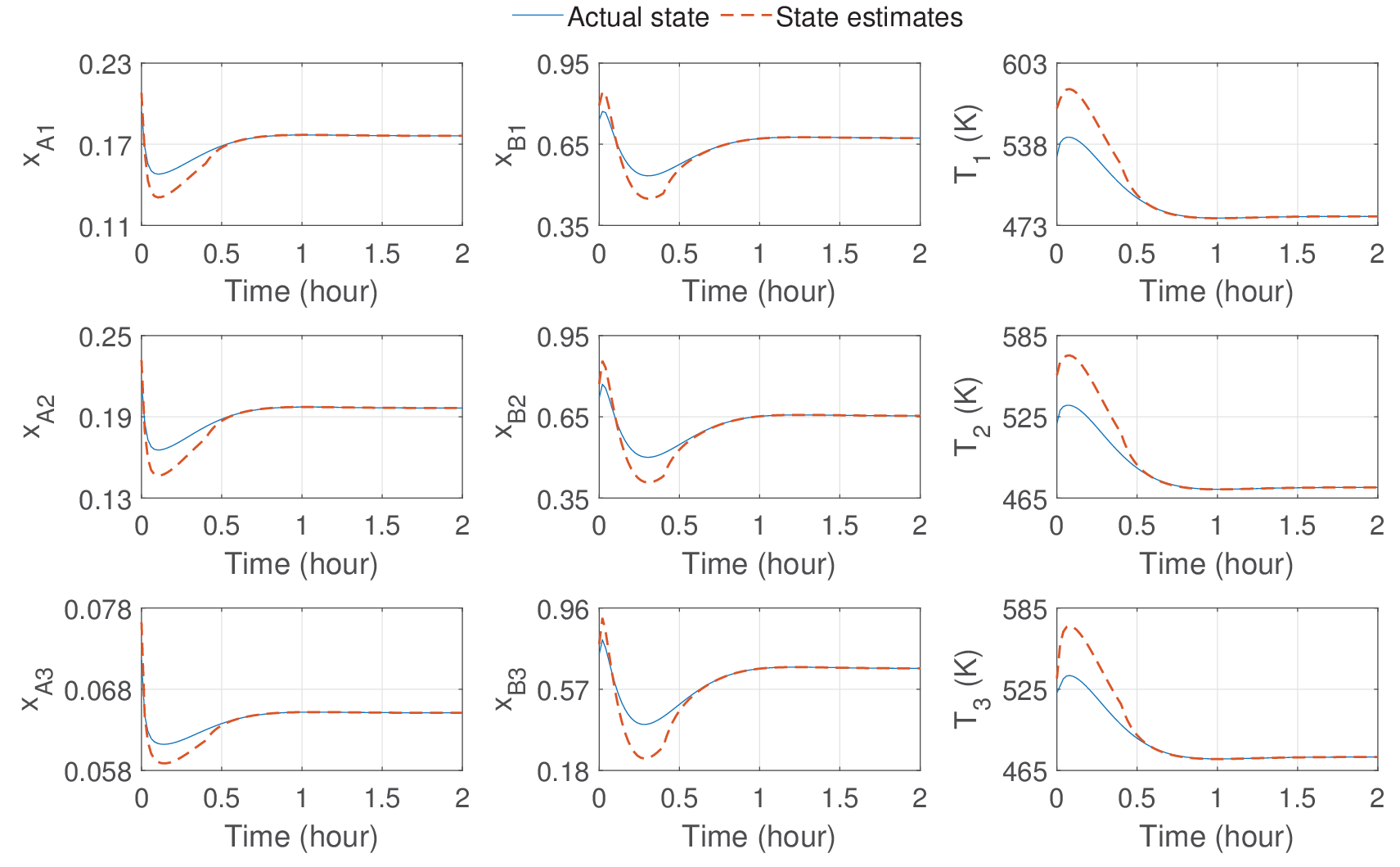}
  \caption{Trajectories of the actual system states and state estimates for the three vessels.}\label{1}
\end{figure}
Subsequently, we compare the estimation performance of the proposed method and the method in \cite{schneider2015convergence} in the presence of system disturbances and measurement noise. The length of the estimation horizon, and the number of iteration steps for each estimator, as well as the weighting matrices $P_{i}$ and $R$ in \cite{schneider2015convergence} are made the same as those for DMHE-1 proposed in this paper.
The performance of each of the two DMHE algorithms is assessed by calculating the RMSE over 500 Monte Carlo runs. 
The trajectories of the average RMSE (calculated based on the 500 runs) for the estimates given by the two DMHE algorithms
are shown in Figure \ref{2}. The overall trends of the two RMSE trajectories are similar. 
At the same time, under this setting, the proposed method provides an overall smaller estimation error as compared to the method in \cite{schneider2015convergence}. This may be because of the incorporation of the penalties on system disturbances into the individual estimators of the distributed scheme in the proposed method.
\begin{figure}[tttt]
  \centering
  \includegraphics[width=0.72\textwidth]{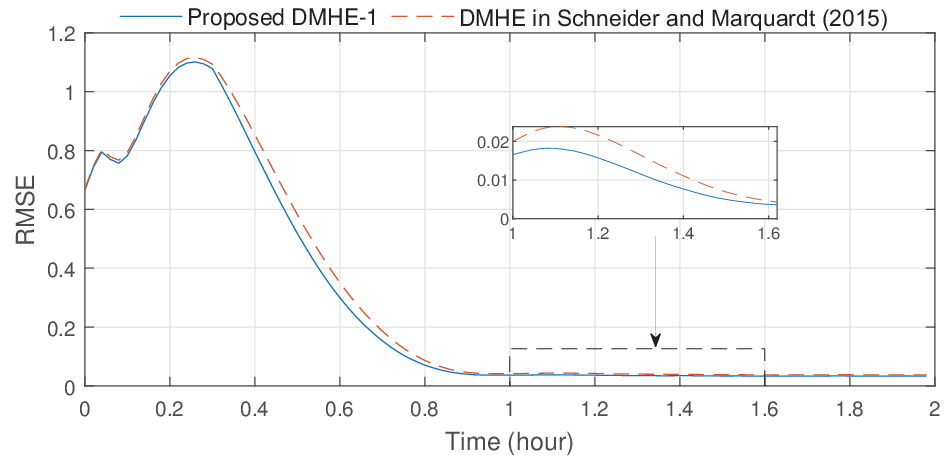}
  \caption{The trajectories of the average RMSEs for the proposed DMHE-1 and the DMHE algorithm in \cite{schneider2015convergence} calculated based on 500 simulation runs with the standard deviation of the system disturbances being $\sigma_w=0.01$.}\label{2}
\end{figure}

We further compare the estimation performance of the proposed DMHE-1 and the algorithm in \cite{schneider2015convergence} under conditions with different magnitudes of system disturbances. Specifically, we consider measurement noise of the same magnitude, while the standard deviation of the system disturbance varies from 0.001 to 0.05. The weighting matrices $Q$ and $R$ are considered diagonal matrices with the main diagonal elements being 0.01. With system disturbances of different standard deviations, we calculate the mean values of the average RMSEs, over 2-hour estimation, for both DMHE algorithms; the results are shown in Figure~\ref{3}. The results show that increasing the magnitude of the system disturbances leads to an increase in the average estimation error for each of the DMHE algorithms, which is consistent with our intuition. In addition, with the chosen estimator parameters, the estimation error of DMHE-1 is slightly smaller than that of the DMHE algorithm in \cite{schneider2015convergence} under all the considered case scenarios with different magnitudes of system disturbances. Figure~\ref{3} also indicates that the difference between the mean values of estimation errors for the two algorithms becomes larger with the increase in the standard deviation of the system disturbances. The results demonstrate the superiority of the proposed method due to the incorporation of system disturbances into the objective functions of local estimators.

\begin{figure}[tttt]
  \centering
  \includegraphics[width=0.65\textwidth]{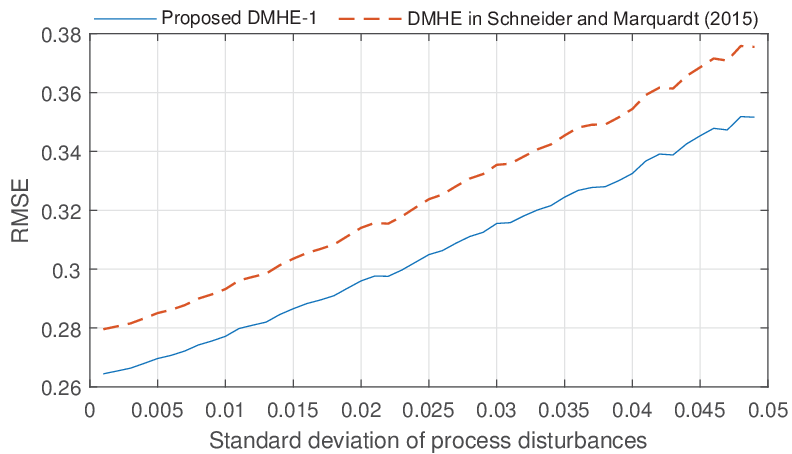}
  \caption{The mean values of the average RMSEs for the proposed DMHE-1 and the DMHE algorithm in \cite{schneider2015convergence} with different magnitudes of the system disturbances.}\label{3}
\end{figure}


\subsection{Results for DMHE-2}
In this subsection, DMHE-2 is applied to handle the hard constraints on the states of the chemical process. In particular, considering the physical meanings of the system states (i.e., material mass fractions and temperatures), the estimates of material mass fractions are required to be bounded within $[0,1]$, and the estimates of temperature are required to be non-negative as DMHE-2 is applied.
Additive disturbances are generated following normal distribution of zero mean and standard deviation of 0.01, and are then made bounded within $[-0.1,0.1]$. Weighting matrix $P_{i}$ for each estimator of DMHE-2 is considered a diagonal matrix with the main diagonal elements being 0.01. Other parameters set are chosen the same as those for DMHE-1. The local estimators for the three subsystems are designed following \eqref{DMHE2}, and Algorithm~\ref{alg1} is implemented to provide estimates of the process states.
Figure \ref{4} presents the trajectories of the actual system states and the estimates provided by the DMHE-2. In the constrained setting, DMHE-2 can provide good estimates of the process states.
The average computation time for one-time execution for each estimator of DMHE-2 with different estimation horizon lengths is computed based on 100 simulation runs. The largest average computation time among the three estimators for DMHE-2 is also presented in Table \ref{tbl:time}.
The computation time increases as the length of the estimation horizon increases. The average computation time for DMHE-2 is significantly larger than that of DMHE-1. This is because DMHE-2 involves solving individual optimization problems in each iteration.
\begin{figure}[tttt]
  \centering
  \includegraphics[width=0.95\textwidth]{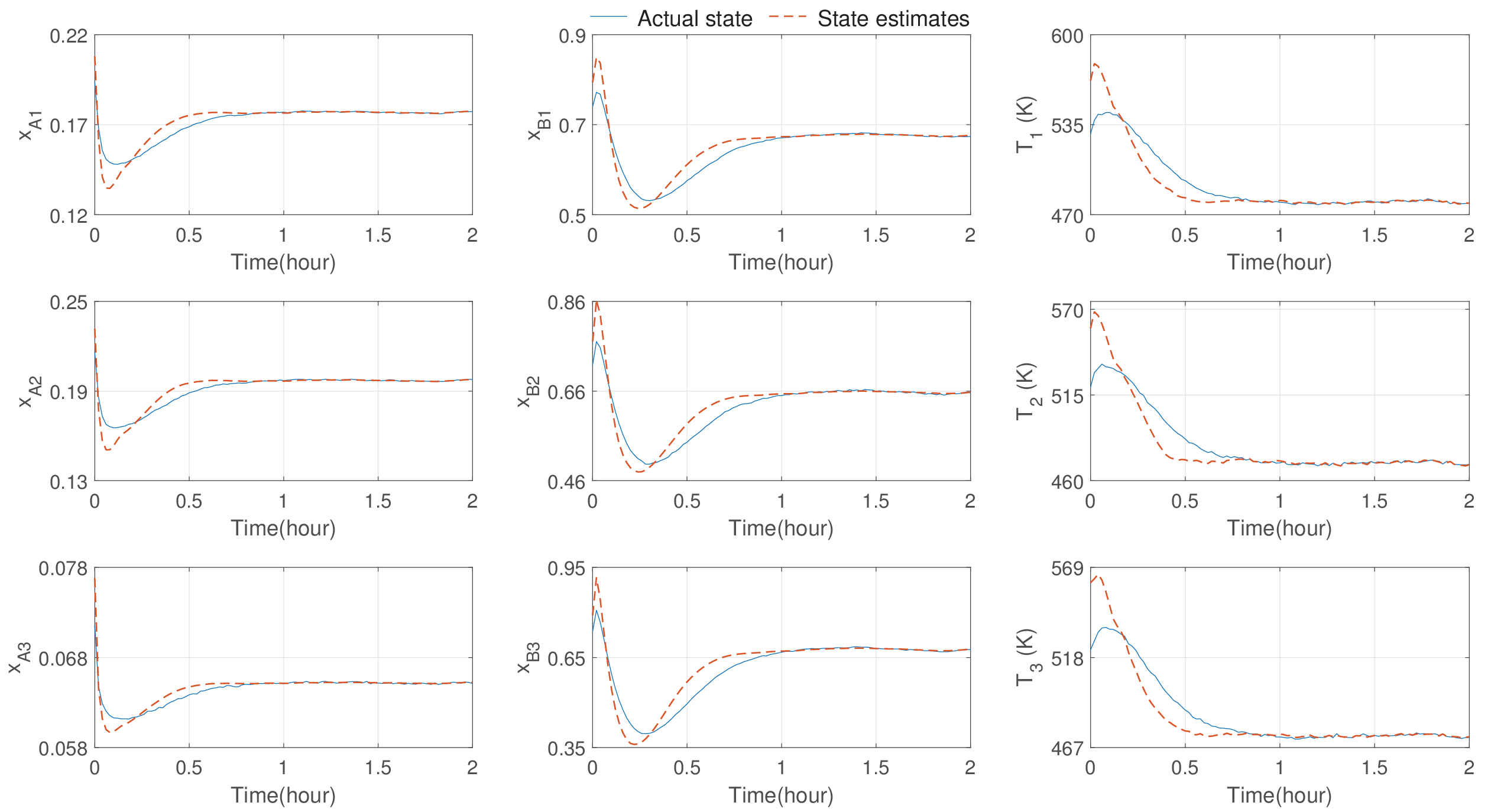}
  \caption{Trajectories of the actual system states and state estimates for three vessels CSTRs with constraints.}\label{4}
\end{figure}
\newpage
\section{Conclusions}
In this work, we revisited the problem of partition-based distributed state estimation of linear systems. The entire process model was decomposed into smaller subsystems that interact with each other. Based on the subsystem models, a distributed state estimation scheme was developed with local estimators designed based on MHE. Both the unconstrained case and the constrained case were taken into account, and two different DMHE formulations were presented for plant-wide state estimation under the two case scenarios. In both formulations, the MHE-based estimators of the distributed scheme are required to be executed iteratively, such that the estimates given by the proposed distributed scheme can be made convergent to the estimates of the corresponding centralized MHE. The convergence of the iterative executions of the MHE estimators, and the stability of the estimation error dynamics were proven. A simulated chemical process consisting of two reactors and a separator was used to illustrate the proposed method. Good estimates of the process state variables were obtained for both DMHE formulations.

\section*{Acknowledgment}

This work is supported by Ministry of Education, Singapore, under its Academic Research Fund Tier 1 (RG63/22), and Nanyang Technological University, Singapore.


%
%
%
%
\newpage
\appendix
\section*{Appendix} \label{sct:appendix}
The system matrices of the discrete-time linearized model for the chemical process in Section \ref{sec.5} are presented below.
\begin{align*}
  A & =  \left[\begin{array}{ccccccccc}
                 0.1401 & -0.0079 & -0.6150 & 0.0925 & -0.0034 & -0.1887 & 0.1978 & -0.0055 & -0.3139 \\
                 0.2102 & 0.3358 & 0.1527 & 0.0394 & 0.1134 & 0.0731 & 0.1076 & 0.2631 & 0.0952 \\
                 0.0395 & 0.0059 & 0.5144 & 0.0135 & 0.0022 & 0.1789 & 0.0298 & 0.0055 & 0.3992 \\
                 0.1269 & -0.0064 & -0.6433 & 0.0802 & -0.0031 & -0.2113 & 0.1673 & -0.0053 & -0.2957 \\
                 0.2529 & 0.3696 & 0.1645 & 0.0580 & 0.1203 & 0.0773 & 0.0882 & 0.2067 & 0.0968 \\
                 0.0423 & 0.0059 & 0.5551 & 0.0155 & 0.0019 & 0.1889 & 0.0302 & 0.0039 & 0.3383 \\
                 0.0660 & 0.0061 & -0.1895 & 0.0464 & 0.005 & -0.0708 & 0.0857 & 0.0110 & -0.0723 \\
                 0.2793 & 0.3550 & -0.0335 & 0.1851 & 0.2450 & -0.0069 & 0.3195 & 0.4402 & 0.0020 \\
                 0.0236 & 0.0055 & 0.4111 & 0.0107 & 0.0038 & 0.2434 & 0.0133 & 0.0074 & 0.3927
               \end{array}\right]\\
  B & = \left[\begin{array}{ccc}
                -0.0154 & -0.0021 & -0.0053 \\
                0.0050 & 0.0010 & 0.0019 \\
                0.0243 & 0.0023 & 0.0106 \\
                -0.0135 & -0.0051 & -0.0042 \\
                0.0047 & 0.0024 & 0.0016 \\
                0.0174 & 0.0098 & 0.0016 \\
                -0.0031 & -0.0013 & -0.0008 \\
                0.0002 & 0.0003 & 0.0001 \\
                0.0076 & 0.0058 & 0.0210
              \end{array}\right]
  C  = \left[\begin{array}{ccccccccc}
                0 & 0 & 1 & 0 & 0 & 0 & 0 & 0 & 0 \\
                0 & 0 & 0 & 0 & 0 & 1 & 0 & 0 & 0 \\
                0 & 0 & 0 & 0 & 0 & 0 & 0 & 0 & 1
              \end{array}\right]
\end{align*}

\end{document}